%% file: RSA_JSP_IEEE_format.tex
\documentclass[11pt,onecolumn]{IEEEtran}

\usepackage{amsthm,amssymb,graphicx,multirow,amsmath,color,amsfonts}
\usepackage[update,prepend]{epstopdf}
\usepackage[noadjust]{cite}
\usepackage[latin1]{inputenc}
\usepackage{tikz}
\usetikzlibrary{arrows,calc}		
\usepackage{bbm} 
\usepackage{pdfpages}
\usepackage{tabulary}
\usepackage{multirow}
\usepackage{comment}
\usepackage{dsfont}
\usepackage{pgf,tikz}
\usepackage{mathrsfs}
\usetikzlibrary{arrows}
\usepackage{subcaption}
\usepackage{bigints}
\usepackage{mathtools}
\newcommand\numberthis{\addtocounter{equation}{1}\tag{\theequation}}
\usepackage{mathtools}
\usepackage{dsfont}
\usepackage{color}

\include{notation}

\allowdisplaybreaks 

\begin{document}
\title{Multilayer Random Sequential Adsorption}
\author{
Priyabrata Parida and Harpreet S. Dhillon
\thanks{Authors are with Wireless@VT, Dept. of Electrical and Computer Engineering, Virginia Tech, Blacksburg, VA (Email: \{pparida, hdhillon\}@vt.edu). The support of the US NSF (Grant ECCS-1731711) is gratefully acknowledged.} 
}

\maketitle

\begin{abstract}
In this work, we present a variant of the multilayer random sequential adsorption (RSA) process that is inspired by orthogonal resource sharing in wireless communication networks. In the one-dimensional (1D) version of this variant, the deposition of overlapping rods is allowed only if they are assigned two different colors, where colors are symbolic of orthogonal resources, such as frequency bands, in communication networks. Owing to a strong spatial coupling among the deposited rods of different colors, finding an exact solution for the density of deposited rods of a given color as a function of time seems intractable. Hence, we propose two useful approximations to obtain the time-varying density of rods of a given color. 
The first approximation is based on the recursive use of the known  monolayer RSA result for indirect estimation of the density of rods for the multilayer version. 
The second approximation, which is more accurate but computationally intensive, involves accurate characterization of the time evolution of the gap density function. This gap density function is subsequently used to estimate the density of rods of a given color. 
We also consider the two-dimensional (2D) version of this problem, where we estimate the time-varying density of deposited circles of a given color as a function of time by extending the first approximation approach developed for the 1D case.
The accuracy of all the results is validated through extensive Monte Carlo simulations.
\end{abstract}


\section{Introduction}\label{sec:intro}
Over the last few decades, random sequential adsorption (RSA) has been a subject of much investigation owing to its importance in a wide range of scientific disciplines such as surface chemistry, condensed matter physics, cellular biology, and photonics.
While the origins of this direction of research can be traced back to 1939 when Flory studied random attachment of atomic groups to a long polymer~\cite{flory1939}, it became popular after R{\'e}ny's ``car parking problem'' was published in 1958~\cite{Renyi1958}.
Subsequently, different variants of the RSA problem have been studied extensively (cf. \cite{talbot2000car} and the references therein).

In this article, we propose a variant of the RSA that is inspired by resource sharing in wireless networks. The broadcast nature of wireless communications is both a blessing and a curse in the design of wireless communications systems. On one hand, it has allowed a large-scale adoption of wireless communications through various radio and television broadcast technologies in which the same wireless signals transmitted by various radio and television stations can be efficiently received by thousands of wireless receivers. On the other hand, the same broadcast nature results in interference when wireless terminals receive signals that are not intended for them. In fact, radio frequency interference is widely regarded as the single most important performance-limiting factor in modern wireless systems. One potential approach to limiting the effect of interference is to ensure a minimum distance among the nodes that are transmitting on the same frequency bands. If the locations of the wireless nodes are modeled as a Poisson point process (PPP), the subset of nodes transmitting on any given frequency will form an RSA process. Furthermore, one can draw parallels between the aforementioned wireless resource allocation and the R{\'e}ny's car parking problem~\cite{coffman1998}. In the same way as a newly arrived car cannot overlap with an already parked car, a newly arriving user in a wireless network cannot be assigned a wireless resource (frequency band in the above discussion) that is already being used in its vicinity. 

\subsection{Motivation}
{To establish a more accurate connection between the RSA process and the above wireless setting, let us first define the canonical {\em monolayer} RSA process. 

\begin{ndef} Consider a $D$-dimensional space, where hard spheres of diameter $\sigma$ appear following a homogeneous spatio-temporal PPP. The monolayer RSA process is constructed by sequentially and irreversibly adding spheres from this space-time PPP with the condition that an arriving sphere does not overlap with already existing spheres in the system.
\end{ndef}

There are two important properties of the RSA process that have been extensively studied in the literature: the time-varying density and the time-varying pair correlation function. 
It is worth noting that due to the irreversible nature of the process, the system eventually reaches a state where the addition of new spheres is not possible. 
This state is known as the {\em jamming state}. Many interesting works, mostly for 1D and 2D versions of the monolayer RSA, are available in the literature. These works highlight some interesting properties of the RSA process related to the transient as well as jamming states. For a pedagogical treatment of the monolayer RSA process, interested readers are advised to refer to~\cite{talbot2000car}.

Given the above definition of the monolayer RSA process, let us now consider a wireless network with $K$-orthogonal frequency bands and assume the following: 
(1) the nodes appear in the network as per a spatio-temporal PPP, 
(2) each node transmits on a certain frequency band, which is randomly selected from the set of {\em available bands}, where a band is said to be available if it is not being used by any other nodes within a certain minimum distance from this node,
(3) a node for which the set of available bands is empty because of the minimum distance  violation (i.e., all bands are already being used by the other nodes in its vicinity) is not admitted into the system, and 
(4) once a node is admitted into the system, it will not leave the system. 
A natural question for this setting is: at any given time what is the density of nodes transmitting on the same frequency band? If we consider a single frequency band ($K=1$) in the network, the setup reduces to the monolayer RSA setting defined above. Hence, it is straightforward to answer this question by leveraging the well-known monolayer RSA results (cf.~\cite{talbot2000car} and the references therein). However, if there are multiple frequency bands, one can envision the resulting point process of nodes as a {\em multilayer RSA}. Owing to the strong spatio-temporal coupling among the layers, the monolayer RSA result cannot be directly applied to study the density of this multilayer RSA process. Further, as will be discussed shortly, even though one can draw some similarities between this multilayer RSA and some known variants of RSA studied in the literature, the underlying physical phenomenon that generates this process has not been discussed in this context yet. Given the novel setting, we naturally need to derive new results to answer the above question, which is the main contribution of our work.}

Before we proceed further, it is important to note that wireless communications research has had many subtle connections with statistical physics over the years. The most relevant to this paper are the ones that emerged because of the use of point processes and stochastic geometry in both areas. For instance, a popular approach to modeling the placement of mobile towers in a wireless cellular network is to view them as a realization of the PPP \cite{Andrews2011,DhiGanJ2012}. Therefore, the service regions of the mobile towers can be modeled as Poisson Voronoi cells with the towers placed at their nuclei. If a user is placed uniformly at random in the typical Poisson Voronoi cell, it is important to know how far is it from the mobile tower (nucleus). The distribution of this distance was derived very recently in \cite{mankar2020} using ideas that were originally developed in statistical physics to study the temporal evolution of the domain structure of the Poisson Voronoi tessellation in \cite{pineda2007}. Along similar lines, the moments of the area of the {\em edge cells} in a bounded Poisson Voronoi tessellation were derived in \cite{koufos2019}. Yet another recent example is the use of line processes for modeling road systems in vehicular communications \cite{dhillon2020}. Line processes have been extensively studied in statistical physics (for instance, to model the trajectories of sub-atomic particles). However, the application of line processes to vehicular communications led to a new problem statement about the path distances measured along the lines, which has also been tackled recently in the statistical physics literature \cite{chetlur2020}. 

With this general background, we are now ready to present our new variant of the RSA process in the canonical one-dimensional (1D) setting below. After deriving results for the 1D case, we will also tackle the two-dimensional (2D) case later in the paper.
Note that we will consider impenetrable hard rods/circles (centered at the arriving points) for the construction of the multilayer RSA process. From the wireless network perspective, the centers represent communicating nodes and the inhibition distances represent their communication range within which they will interfere with other nodes transmitting on the same frequency band. Our problem setup implicitly assumes that the interference is modeled using a generic distance-dependent path-loss. However, if one also considers random fluctuations in the interference power due to other wireless channel impairments, such as shadowing and small-scale fading, it will result in other interesting variants of the RSA process. We will present a brief discussion on this topic in Sec.~\ref{sec:concl}.

\subsection{Problem statement}\label{sec:ProbState}

Consider a 1D line that is empty at $t=0$. Hard rods of length $\sigma$ are arriving uniformly at random at rate $r_a$ per unit length. A rod is placed on the line irreversibly after being assigned a color from a set of colors ${\cal K} = \{c_1, c_2, \ldots, c_K\}$. From a communication network perspective, rod centers represent communicating node locations, their lengths represent the communication range, and the set colors represent the orthogonal frequencies. A color is selected randomly from the set of {\em available colors}, where a color is available if it is not assigned to already existing rods that overlap with the arriving rod.
If no colors are available for assignment, the arriving rod is not admitted into the system. 
An illustrative diagram is presented in Fig.~\ref{fig:ProbState}.
The special case of $K=1$ gives us the celebrated {\em R{\'e}ny's car parking} problem.
Our goal is to characterize the density of rods of a given color as a function of time, denoted by $\rho_k(t)$ for $c_k \in {\cal K}$. 

\begin{figure}[!htb]
  \centering
  \includegraphics[width=0.8\columnwidth]{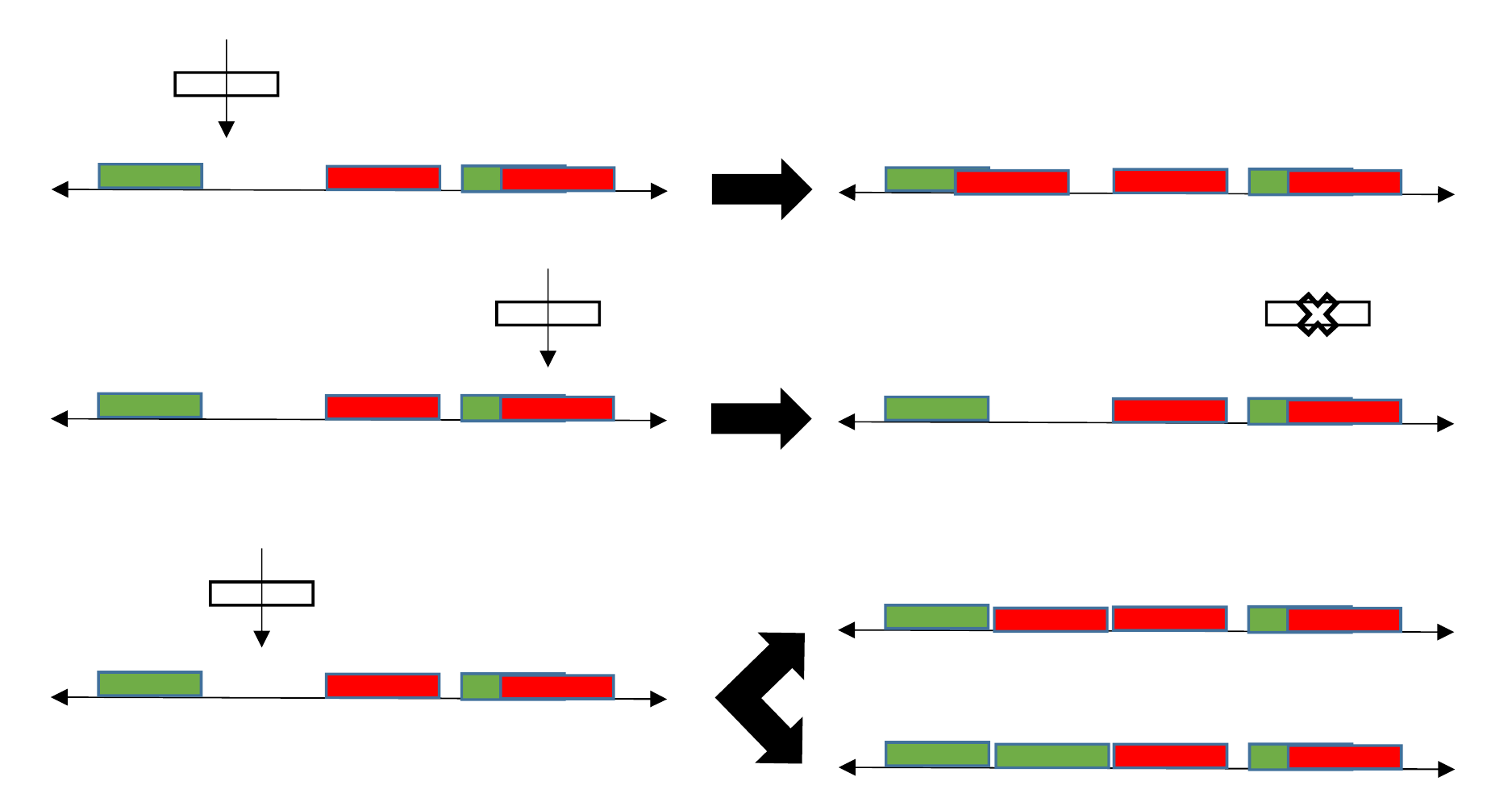}
  \caption{\footnotesize An illustrative figure for the deposition of rods that are assigned either red or green color. (Top) Arriving rod overlaps with a deposited rod of green color. Hence, it is assigned red color. (Middle) Arriving rod overlaps with rods of both the colors. Hence, it is discarded. (Bottom) Arriving rods lies in an empty interval. Hence, it can be assigned either of the two colors with equal probability.}
  \label{fig:ProbState}
\end{figure}

\subsection{Related literature and contributions}
The problem described above has similarity to some known multilayer variants of the RSA problem~\cite{bartelt1990,bartelt1991,krapivsky1992,van1997,yang1998}. While we briefly describe these variants next for completeness, the differences in the geometric constraints and dynamics between these and the setup studied in this paper forbid a direct application of these prior analyses to the current setup.
In \cite{bartelt1990,bartelt1991}, authors have considered a sequential multilayer deposition of  dimers on a lattice. Using mean-field theory, approximate density results are presented by not considering the {\em screening}\footnote{Blocking of arriving objects by higher layers to the lower layers due to overhangs. This phenomenon is not a characteristic of our model due to orthogonal frequency bands.} effect from higher layers. Additional approximate results for the entire time range based on {\em empty interval probability}\footnote{The probability of finding an interval of $n$ or more consecutive sites empty.} rate equations were also presented. However, the results are limited to the first two layers as the solution rapidly becomes cumbersome for higher layers. Further, in \cite{bartelt1991} authors provide the large time asymptotic behavior of the densities for different layers for the continuum case. 
In~\cite{krapivsky1992}, the authors present asymptotic results for a variant of the multilayer RSA with sequential deposition of objects without screening effect. In addition, the authors consider the length of the objects to be random with a certain distribution. The asymptotic results are presented for 1D and 2D continuum cases that suggest each layer approaches the jamming limit as a power law. 
In~\cite{van1997}, the authors study a variant of the continuum multilayer RSA where variable-length screening due to overhangs from higher layers is considered. For this model, exact results are presented only for the first layer. 
A generalized version of the multilayer RSA model of~\cite{van1997} is considered in~\cite{yang1998} where the three possible events of the particle deposition are taken into account namely adsorption, desorption, and rolling of an object on the surface.  
Similar to the previous case, the exact results are presented only for the first layer.

Another interesting line of works that are inspired by the process of frequency assignment in wireless networks can be found in~\cite{Fleurke2007,fleurke2009second,fleurke2010multilayer,Fleurke2014}. In this variant, the sequential assignment of frequencies gives rise to a space-time process that is similar to the multilayer RSA process without the screening effect. 
In~\cite{Fleurke2007}, through numerical simulations, authors propose several conjectures related to the long term asymptotic behavior such as the number of frequency bands necessary to accommodate $n$ users, i.e., the average number of layers formed by deposition of the first $n$ rods. Additional simulation-based results related to packing density are also presented. 
Inspired by the same model, in~\cite{fleurke2009second} a sequential two-layer RSA process is considered on a discrete finite lattice where the arriving objects are dimers. The model also takes into account both ``no screening'' and screening of dimers from the second layer to the first layer.
The density results are presented for local patterns,  i.e. occupancy in both the layers over three consecutive sites.  
In~\cite{Fleurke2014}, authors extend the previous results from two layers to higher layers for a finite lattice size of five sites where arrival is allowed on consecutive three sites. 
The analysis is focused on obtaining the occupancy probability of the center site for a given layer at the large time limit. Further, a few simulation-based results for systems with larger lattice sizes are also discussed.

From this discussion, two key characteristics of the prior works are noteworthy. First, each variant of the multilayer RSA has unique geometrical and dynamical features that are not universal and are strongly driven by the underlying rules of deposition of the objects. 
Because of this, a unified analysis of all these variants, although desirable, is not possible. As a result, understanding the characteristics of each process requires a unique analytical treatment governed by its underlying physical model. Second, the exact characterization of these features is extremely difficult due to the non-markovian nature of the process as well as strong spatio-temporal interaction among different layers. Hence, accurate approximate results are mostly our best hope unless one considers very specific limiting scenarios, such as finite lattice size or large time system behavior.
With this understanding, the contributions of our work are summarized below: 
\begin{enumerate}
\item as described in Sec.~\ref{sec:ProbState}, we propose a new variant of the multilayer RSA that is inspired from random orthogonal resource sharing in wireless communications networks. 

\item Although each step in this variant is random, owing to the infinite memory of the deposition process, it is non-markovian. Hence, obtaining exact results for the kinetics is difficult. Therefore, to tackle this problem, we develop approximations that are reasonably accurate for the entire time range.  For the 1D case, we provide two useful approximation methods to obtain the density of rods of a given color. The first method recursively uses the monolayer RSA result with modified arrival rates to obtain the density of rods of a given color. On the other hand, in the second method, we approximately characterize the {\em gap density function}, which is later used to obtain the density of rods. While the first approach is more amenable to numerical evaluation, the second method is more accurate along with providing useful intermediate results.

\item We also accommodate the 2D version of the problem, which is solved using a method that is similar to the first approximation method for the 1D case. 
{From an application point of view, we present a case study of orthogonal frequency band allocation in Wi-Fi networks where the results derived for the 2D RSA are directly applicable for the system analysis.}
\end{enumerate}

The rest of the article is organized as follows: in Sec.~\ref{sec:Apprx2_1D}, we present the first approximation that leverages the monolayer RSA result to solve the problem. In Sec.~\ref{sec:Apprx1_1D}, we present our second approach to solve the problem using gap density function. 
The density results for the 2D case using the first approximation is presented in Sec.~\ref{sec:Apprx2_2D}.
We provide concluding remarks in Sec.~\ref{sec:concl}.

\section{Density Approximation of 1D Multilayer RSA: An Iterative Approach}\label{sec:Apprx2_1D}

In this section, we present our first approach to approximate the density of rods of a given color as a function of time. This approach is based on establishing an equivalence between the proposed color assignment process and an alternate sequential color assignment process that is described below. The equivalence between these two assignment processes is in terms of total density of rods admitted into the system.

The rules for the alternate sequential color assignment scheme are as follows: 
\begin{enumerate}
\item Let there be $K$ colors $\ncalK = \{c_1, c_2, \ldots, c_K\}$ with a predefined ordering. The coloring scheme is sequential, i.e. for an arriving rod at $x$, color $c_1$ is considered first. If a rod of color $c_1$ overlaps with $B_{\sigma/2}(x)$\footnote{Throughout the manuscript, we denote a ball of radius $\sigma/2$ centered at $x$ as $B_{\sigma/2}(x)$.}, then color $c_2$ is considered and so on. 

\item If the arriving rod at $x$ overlaps with rods of all the colors, i.e. centers of rods of all colors are present in $B_{\sigma}(x)$, then the rod is not admitted into the system. 
\end{enumerate}

Let $\tilde{\rho}_i(t)$ be the density of rods of color $c_i$ at time $t$.
Due to the sequential nature of the assignment scheme, it is clear that $\tilde{\rho}_1(t) \geq \tilde{\rho}_2(t) \geq \ldots \geq \tilde{\rho}_K(t)$.
On the other hand, in case of the random assignment of colors as proposed in the original problem (Sec~\ref{sec:ProbState}), the densities of rods of different colors are the same, i.e. $\rho_1(t) = \rho_2(t) = \ldots = \rho_K(t)$. 
Note that at time $t$, in both the schemes, the total density of admitted rods of all colors is the same. 
Hence, we write
\begin{align*}
& \sum_{k=1}^K \rho_k(t) = \sum_{k=1}^K \tilde{\rho}_k(t) \\
\Rightarrow & K \rho_i(t) = \sum_{k=1}^K \tilde{\rho}_k(t) \\
\Rightarrow & \rho_i(t) = \frac{\sum_{k=1}^K \tilde{\rho}_k(t)}{K}, \forall i = 1, 2, \ldots, K. \numberthis
\label{eq:Den_Iter_App}
\end{align*} 

To use the above equation to characterize the density of rods of color $c_i$ for the original assignment scheme, we need information regarding $\tilde{\rho}_i(t), \ \forall i$.
Observe that the evolution of density of rods for color $c_1$, denoted by $\tilde{\rho}_1(t)$, is the same as the monolayer RSA. 
Hence, the density of rods of color $c_1$ is given as~\cite{talbot2000car}
\begin{align}
\tilde{\rho}_1(t) = \frac{1}{\sigma} \int_0^{r_a \sigma t} \exp\left(-2 \int_0^u \frac{1-e^{-x}}{x} {\rm d}x\right) {\rm d}u.
\end{align}
However, characterizing the exact density of rods of color $c_n$, for $n \geq 2$, is non-trivial. 
Hence, we approximate $\tilde{\rho}_n(t)$  for $n \geq 2$.
In the sequential assignment scheme, at time $t$, $\tilde{\rho}_1(t)$ rods per unit length have been assigned color $c_1$. 
Hence, the number of arrivals per unit length that have been considered for the allocation of color $c_2$ is $r_a t - \tilde{\rho}_1(t)$.
Similarly, the number of arrivals per unit length considered for color $c_3$ is $r_a t - \tilde{\rho}_1(t) - \tilde{\rho}_2(t)$.
To obtain $\tilde{\rho}_2(t)$, we assume that the rods are arriving uniformly at random at a rate $r_a - \frac{\tilde{\rho}_1(t)}{t}$. Note that although reasonable, this assumption is an approximation.
Further, assuming that the evolution of color $c_2$ happens similar to monolayer RSA, the density at time $t$ is given as
\begin{align}
\tilde{\rho}_2(t) = \frac{1}{\sigma} \int_0^{r_a \sigma t - \tilde{\rho}_1(t) \sigma} \exp\left(-2 \int_0^u \frac{1-e^{-x}}{x} {\rm d}x\right) {\rm d}u.
\end{align}
Proceeding on the similar lines, the density of rods of color $c_n$ for $2 \leq n \leq K$ is given as 
\begin{align}
\tilde{\rho}_n(t) = \frac{1}{\sigma} \int_0^{r_a \sigma t - \sigma \sum_{i=1}^{n-1}\tilde{\rho}_i(t)} \exp\left(-2 \int_0^u \frac{1-e^{-x}}{x} {\rm d}x\right) {\rm d}u.
\end{align}

In the following proposition, we summarize the density result presented in this section: 
\begin{prop}
The density of rods of a given color for the original random color assignment problem is given as 
\begin{align*}
\rho_i(t) = \frac{\sum_{k=1}^K \tilde{\rho}_k(t)}{K},
\end{align*}
where 
\begin{align*}
\tilde{\rho}_k(t)= \frac{1}{\sigma} \int_0^{r_a \sigma t - \sigma \sum_{i=1}^{k-1}\tilde{\rho}_i(t)} \exp\left(-2 \int_0^u \frac{1-e^{-x}}{x} {\rm d}x\right) {\rm d}u.
\end{align*}
\end{prop}

The validation of the accuracy of the above approximation is presented in Fig.~\ref{fig:Density_Iterative_MLRSA}. 

\begin{figure}[!htb]
  \centering
  \includegraphics[width=0.6\columnwidth]{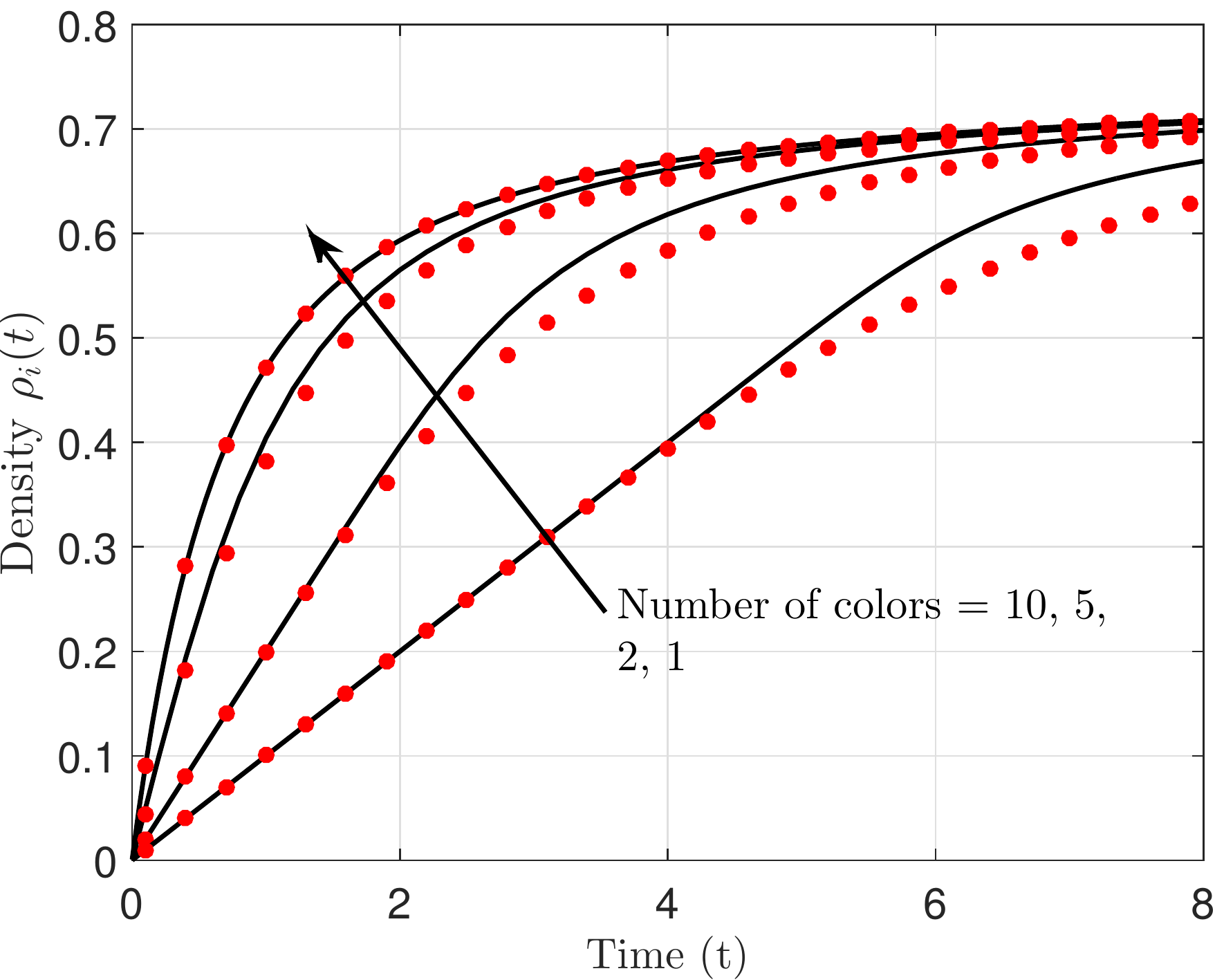}
  \caption{\footnotesize The evolution of the density of rods of a particular color as a function of time $t$ for $\sigma = 1$. Markers and solid lines represent simulations and theoretical results, respectively.}
  \label{fig:Density_Iterative_MLRSA}
\end{figure}

Interestingly, once the equivalence between the original and the alternate sequential color assignment processes was established in~\eqref{eq:Den_Iter_App}, this approach relied exclusively on the known monolayer result. As we will discuss in Section 4, its tractability also makes it an appealing choice for the RSA analysis in higher dimensions. That said, this approach suffers from a gradual loss of accuracy as the number of colors increases. This motivates us to present an alternate result that is more accurate compared to this approximation and has an added advantage of providing useful intermediate results that have more information regarding the kinetics of the process (whereas the above approach does not provide any other statistical information about the original random color assignment process apart from the time-varying density of rods of a given color).

\section{Density Approximation for 1D Multilayer RSA: Gap Density Function-based Approach}\label{sec:Apprx1_1D}

In this section, we present our second approximation approach to obtain the density of rods of a given color.
It is based on the characterization of the {\em gap density function}, which is one of the canonical methods to understand the kinetics of the RSA process as well as its different variants.
In our case, at time $t$, the gap density function $G_i(l,t)$ is defined such that $G_i(l,t) {\rm d}l$ gives the {density} of gaps of length  between $l$ and $l + {\rm d}l$ for rods that are colored $c_i$.
Following properties of $G_i(l,t)$ are useful in the derivation of density of rods of a given color:
\begin{enumerate}
\item Since each gap corresponds to an admitted rod of color $c_i$ preceding it (or succeeding it), the density of rods of color $c_i$ is given as 
\begin{align*}
\rho_i(t) = \int_0^{\infty} G_i(l, t) {\rm d}l. \numberthis
\label{eq:CarDensity_Gt}
\end{align*}
This direct relationship to the density makes gap density function more attractive to work with compared to other intermediate quantities such as empty interval probability~\cite{bonnier1994}.

\item At time $t$, the fraction of the length (average length over a unit interval) available for admitting a rod that can be assigned color $c_i$ is 
\begin{align*}
\Phi_i(t) = \int_\sigma^{\infty} (l-\sigma) G_i(l, t) {\rm d}l. \numberthis
\label{eq:Phit_Gt_Reln}
\end{align*}
Above result can be interpreted as the probability of a rod arriving in a gap of length $l$ of color $c_i$. This relationship is used later in the proposed approximation.

\end{enumerate}


Instead of directly solving the problem for $K \geq 2$ colors, we begin with the simpler case of $K=2$. The objective is to expose the underlying structure of the problem for the simpler setting of $K=2$, which will help in identifying key constructs that emerge from the inherent spatial coupling of the RSA and will hence need careful approximations for a tractable analysis. This will then inform our analysis of $K\geq 2$. 

\subsection{Results for two layers $(K=2)$}

Consider the scenario where rods can be assigned either of the two colors ${\cal K} = \{c_1, c_2\}$. As mentioned in Sec.~\ref{sec:ProbState}, the assignment of a color is random with equal probability unless the arriving rod overlaps with an admitted rod of a given color. 
Owing to the random assignment, at a given time $t$, $G_1(l, t)$ and $G_2(l, t)$ are identical. 
Hence, without loss of generality, we just focus on deriving $G_1(l, t)$.

Our first step is to characterize the time evolution of $G_1(l, t)$. Consider a gap of length $l$ for rods of $c_1$ (see Fig.~\ref{fig:GapLineSeg}). The allowable length on which a rod can arrive with the possibility of getting the color $c_1$ is the segment $[\frac{\sigma}{2}, l - \frac{\sigma}{2}]$. Let us denote this line segment by $L_{l-\sigma}$.
For a rod arriving at location $x \in L_{l-\sigma}$, we define the following events: 
\begin{enumerate}
\item ${\cal I}_i(x, t):=$ \{A rod arriving at point $x$ during the time window $(t, t + {\rm d}t]$ will be assigned $c_i$.\}
\item ${\cal C}_i(x, t, l):=$ \{The rod arrives in  gap of length $l$ corresponding to color $c_i$ during time $(t, t+{\rm d}t]$.\}
\item ${\cal E}_{n}(x, t):=$ \{At time $t$, the segment $B_{\sigma/2} (x) \coloneqq \left[x-\frac{\sigma}{2}, x + \frac{\sigma}{2}\right]$ overlaps with $n$ deposited rods\}
\end{enumerate}

\begin{figure}[!htb]
  \centering
  \includegraphics[width=1\columnwidth]{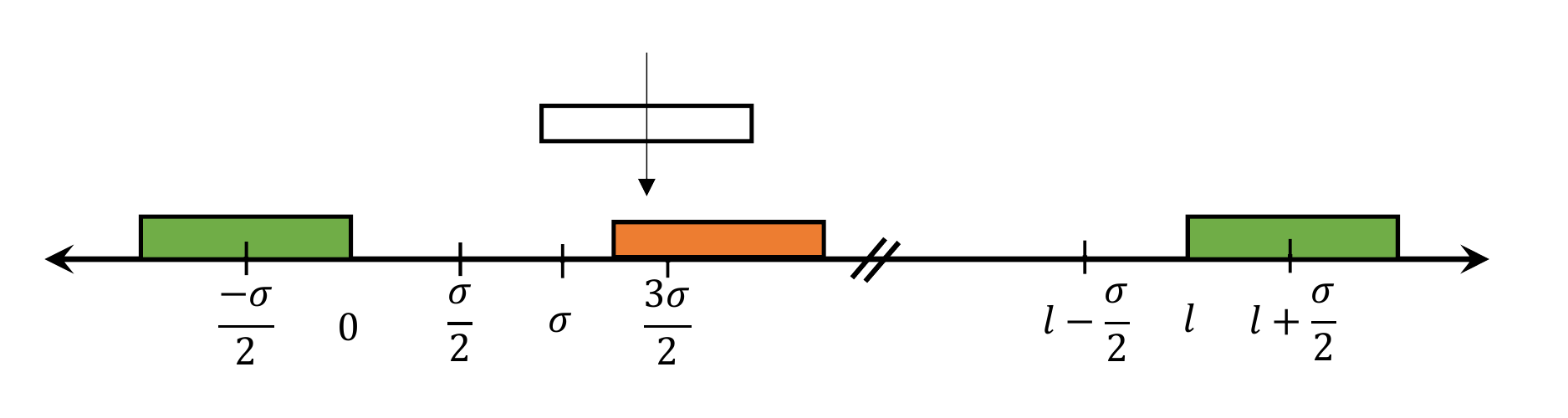}
  \caption{\footnotesize A gap of length $l$ for rods of color $c_1$ (green for the illustration purpose). New arrivals that can destroy this gap are possible only over the segment $L_{l - \sigma} = [\frac{\sigma}{2}, l-\frac{\sigma}{2}]$.}
  \label{fig:GapLineSeg}
\end{figure}

The evolution of $G_1(l,t)$ depends on the following configurations in the vicinity of $x$: 
\begin{enumerate}
\item the rod will be assigned color $c_1$ (green in the illustrations) with probability $1/2$ if $B_{\sigma/2}(x)$ is not partially (or fully) covered by rods of color $c_2$ (orange in the illustrations), and
\item the rod will be assigned color $c_1$ with probability $1$ if $B_{\sigma/2}(x)$ is partially (or fully) covered by rods of color $c_2$. In the illustrative example of Fig.~\ref{fig:GapLineSeg}, the arriving rod will be assigned color $c_1$ with probability~1 as it overlaps with rod of color $c_2$.
\end{enumerate}

We write the following set of differential equations to capture the evolution of $G_1(l, t)$ due to an arrival in the gap of length $l$ for rods of color $c_1$ during an infinitesimally small time window $(t, t + {\rm d}t]$: 
\begin{align*}
\frac{\partial G_1(l, t)}{\partial t} =   
\begin{dcases}
   - r_a \int_{x \in L_{l-\sigma}} G_1(l, t)  \mathbb{P}[{\cal I}_1(x, t) | {\cal C}_1(x, t, l)] {\rm d}x &  \\
   \quad + 2 r_a \int_{y = l+ \sigma}^{\infty} G_1(y, t) \mathbb{P}[{\cal I}_1(x, t) | {\cal C}_1(x, t, y)] {\rm d}y & l \geq \sigma,\\
       2 r_a \int_{y = l+ \sigma}^{\infty} G_1(y, t) \mathbb{P}[{\cal I}_1(x, t) | {\cal C}_1(x, t, y)] {\rm d}y & l < \sigma.
\end{dcases}
\numberthis
\label{eq:Gltdt_BC}
\end{align*}

To obtain \eqref{eq:Gltdt_BC}, we first consider the case of $l \geq \sigma$ as it involves the rate of change in the density of gaps between $(l , l + {\rm d}l]$ due to the destruction of such gaps as well as creation of such gaps from gaps of larger length.
The second case of $l < \sigma$ involves only the creation term that can be obtained using a similar logic as we present for $l \geq \sigma$. 
The first term on the right hand side (destruction term) captures the rate of change in density $G_1(l, t)$ due to the average number of arrivals over unit length in a gap of length $l$ and is straightforward to obtain.
The second term (creation term) captures the rate of change in $G_1(l, t)$ due to average number of arrivals per unit length that can create a gap of length $l$ from a gap of length $y > l + \sigma$.
This expression can be derived as follows: for all the gaps of length $(y, y + {\rm d}l]$, the fraction of available length for arrival of a rod is $(y - \sigma) G_1(y, t){\rm d}l$. In order to create a gap of length $l$, the rod needs to arrive on a thin length ${\rm d}y$ at a distance $l + \sigma/2$ from either end of the gap $y$. Due to the uniform arrival of rods, the probability of this event is ${\rm d}y/(y - \sigma)$.
Further, this arriving rod will be assigned color $c_1$ with certain probability depending  on the configuration of already deposited rods of color $c_2$ in this gap. This probability is captured by the term $\nbbP[{\cal I}_1(x, t)|{\cal C}_1(x, t, y)]$.
Hence, the fraction of length that allows an arriving rod to partition $y$ into two smaller gaps of lengths $l$ and $y-l-\sigma$ is 
\begin{align*}
& \frac{2 {\rm d}y}{(y - \sigma)} (y-\sigma) [G_1(y,t) {\rm d}l] \nbbP[{\cal I}_1(x, t)|{\cal C}_1(x, t, y)] \\
& \quad \quad \quad \quad = 2 [G_1(y, t) {\rm d}l] \nbbP[{\cal I}_1(x, t)|{\cal C}_1(x, t, y)] {\rm d}y,
\end{align*}
which gives the desired integrand in \eqref{eq:Gltdt_BC} for the creation terms in both the cases.

Our next step is to derive an expression for the probability term presented in \eqref{eq:Gltdt_BC}. 
Using Bayes' theorem and law of total probability, we write $\nbbP[{\cal I}_1(x, t)|{\cal C}_1(x, t, l)]=$
\begin{align*}
& \frac{\nbbP[{\cal I}_1(x, t), {\cal C}_1(x, t, l)]}{\nbbP[{\cal C}_1(x, t, l)]} \\
= & \frac{\sum_{n\geq 0} \nbbP[{\cal I}_1(x, t), {\cal C}_1(x, t, l)| {\cal E}_{n}(x,t)] \nbbP[{\cal E}_{n}(x,t)]}{\nbbP[{\cal C}_1(x, t, l)]} \\
= & \frac{\sum_{n\geq 0} \nbbP[{\cal I}_1(x, t)| {\cal C}_1(x,t,l), {\cal E}_{n}(x,t)] \nbbP[{\cal C}_1(x,t,l)| {\cal E}_{n}(x,t)] \nbbP[{\cal E}_{n}(x,t)]}{\nbbP[{\cal C}_1(x,t,l)]}. \numberthis
\label{eq:Prob_I1x_Cond_C1x}
\end{align*}

Note that the above conditional probability depends on the location of $x \in L_{l-\sigma}$. 
Deriving an exact expression while considering this location dependence is intractable. This is a manifestation of the spatial coupling because of which exact analyses of multilayer RSA in most settings is intractable. 
Next we present our approximation approach that is based on a few assumptions including the location independence.

\begin{figure*}[!htb]
\centering
\begin{subfigure}{1\textwidth}
  \centering
  \includegraphics[width=1\linewidth]{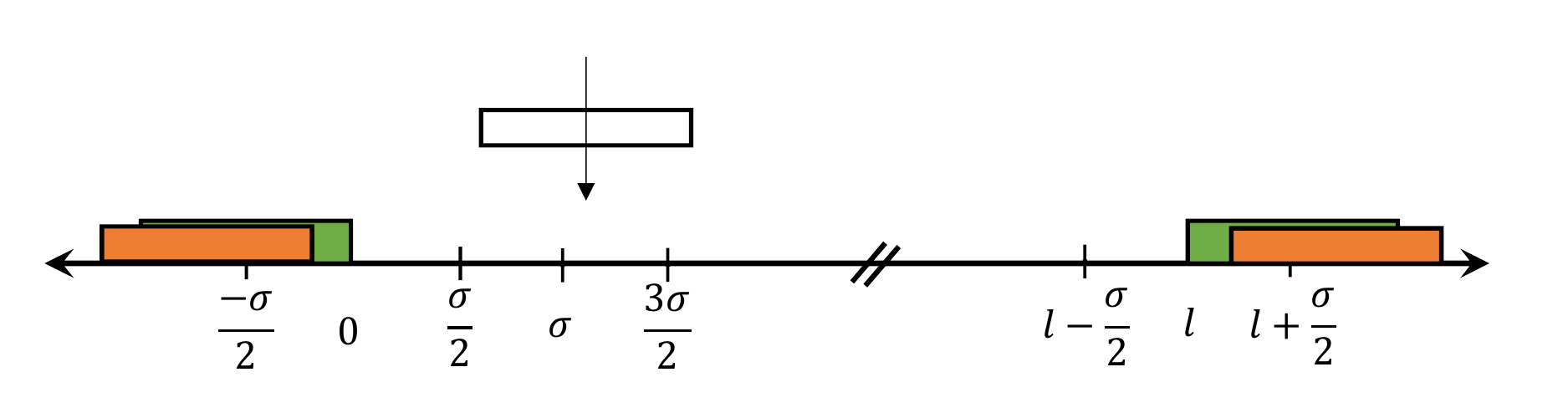}
\end{subfigure}
\begin{subfigure}{1\textwidth}
  \centering
  \includegraphics[width=1\linewidth]{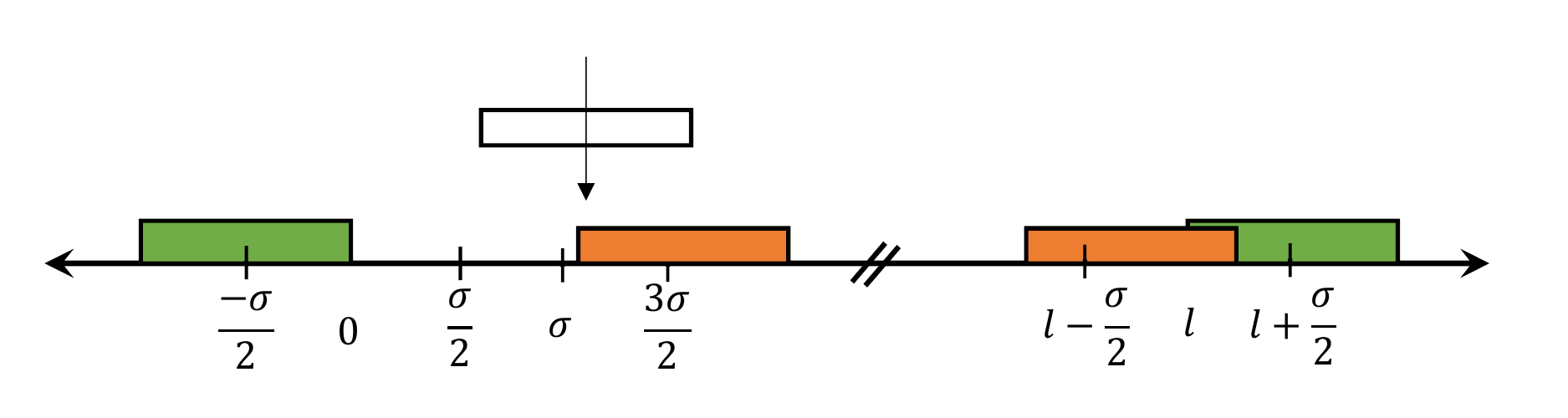}
\end{subfigure}%
\caption{\footnotesize  An illustrative gap of length $l$ for color green ($c_1$). The arrivals at $ x \in \left[\frac{\sigma}{2}, l-\frac{\sigma}{2}\right]$ are considered for assigning color green.}
\label{fig:LineSeg}
\end{figure*}

{
First, we get $\nbbP[{\cal E}_{n}(x,t)]$, i.e. the probability of the event that the interval $[x - \frac{\sigma}{2}, x + \frac{\sigma}{2}]$ overlaps with $n$ deposited rods.
Consider the following realizations of this event:
\begin{enumerate}
\item For $n=0$ (Fig.~\ref{fig:LineSeg} top): if the arrival occurs at $x \in [\sigma, l-\sigma]$, then it is clear that there has been no prior arrivals in $(x-\sigma, x + \sigma)$ until time $t$. Otherwise, it would have been assigned one of the colors. Hence, using empty interval probability of 1D Poisson process, $\nbbP[{\cal E}_{0}(x,t)] = e^{-r_a 2\sigma t}$. On the other hand, if the arrival occurs at $x \in\{[\frac{\sigma}{2}, \sigma) \cup (l-\sigma, l-\frac{\sigma}{2}]\}$, then there is a non-zero probability that there has been atleast one arrival in $(x-\sigma, x + \sigma)$ prior to time $t$. This arrival(s) has been discarded as there are no colors left to assign. Exact evaluation of the probability of this event is cumbersome. Hence, we approximate it as a Poisson arrival and write
\begin{align*}
\nbbP[{\cal E}_{0}(x,t)] = e^{-r_a 2\sigma t}, \quad x \in \bigg[\frac{\sigma}{2},  l-\frac{\sigma}{2}\bigg].
\end{align*}

\item For $n=1$ (Fig.~\ref{fig:LineSeg} bottom): similar to the previous case, if the arrival occurs at $x \in [\sigma, l-\sigma]$, then it is clear that there is one arrival in $(x-\sigma, x + \sigma)$ until time $t$. Hence, we write $\nbbP[{\cal E}_{1}(x,t)] =r_a 2\sigma t e^{-r_a 2\sigma t}$. However, when arrival occurs at $x \in\{[\frac{\sigma}{2}, \sigma) \cup (l-\sigma, l-\frac{\sigma}{2}]\}$, it is difficult to derive $\nbbP[{\cal E}_1(x,t)]$ as it requires a cumbersome enumeration. To circumvent this, similar to the previous case, we approximate the arrivals in $(x-\sigma, x + \sigma)$ to follow a Poisson process and write 
\begin{align*}
\nbbP[{\cal E}_{1}(x,t)] =r_a 2\sigma t e^{-r_a 2\sigma t}.
\end{align*}

\item For $n=2$: Similar to the previous cases, we approximate that the process is Poisson in $(x-\sigma, x + \sigma)$ for an arrival at $x \in [\frac{\sigma}{2}, l - \frac{\sigma}{2}]$. Hence, 
\begin{align*}
\nbbP[{\cal E}_{2}(x,t)] = \frac{(r_a 2\sigma t)^2}{2} e^{- r_a 2\sigma t}.
\end{align*}
\end{enumerate}
Note that $\nbbP[{\cal E}_{n}(x,t)] = 0$ for $n \geq 3$.
}

{
Next, we are interested in $\nbbP[{\cal C}_1(x, t, l)| {\cal E}_n(x, t)], \ \forall n$.
Let us define the event $C_1(x, t)$ as the event that an arriving rod falls in a gap corresponding to color $c_1$.
As presented earlier, the probability of this event is given as 
\begin{align*}
\nbbP[C_1(x, t)] = \Phi_1(t) = \int_\sigma^{\infty} (z-\sigma) G_1(z, t) {\rm d}z.
\end{align*}
Above probability takes into account all the gaps of length greater than $\sigma$, where the probability that the rod lies in a gap of length $(l, l + {\rm d}l]$ corresponding to color $c_1$ is $(l-\sigma) G_1(l, t) {\rm d}l$.
{Please note that $\nbbP[{\cal C}_1(x, t, l)| {\cal E}_n(x, t)] = \nbbP[{\cal C}_1(x, t, l), {\cal C}_1(x, t)| {\cal E}_n(x, t)]$ due to the fact that ${\cal C}_1(x, t, l) \subseteq {\cal C}_1(x, t)$ conditioned on ${\cal E}_n(x, t)$.
Further, we assume that $\nbbP[{\cal C}_1(x, t, l)| {\cal C}_1(x, t), {\cal E}_n(x, t)] = \nbbP[{\cal C}_1(x, t, l)| {\cal C}_1(x, t)]$ for all $n$.
Using this relationship, for $n=0$, this conditional probability is simply the probability that $B_{\sigma/2}(x)$ lies in a gap of length $(l, l + {\rm d}l]$ of all the gaps and is given as
\begin{align*}
\nbbP[{\cal C}_1(x, t, l)| {\cal E}_0(x, t)]  & = \nbbP[{\cal C}_1(x, t, l), {\cal C}_1(x, t)| {\cal E}_0(x, t)] \\ & = \nbbP[{\cal C}_1(x, t, l)| {\cal C}_1(x, t), {\cal E}_0(x, t)] \nbbP[{\cal C}_1(x, t)| {\cal E}_0(x, t)]  \\ 
& \stackrel{(a)}{=} \nbbP[{\cal C}_1(x, t, l)| {\cal C}_1(x, t)] \nbbP[{\cal C}_1(x, t)| {\cal E}_0(x, t)] \\
& = \frac{\nbbP[{\cal C}_1(x, t, l)]}{\nbbP[{\cal C}_1(x, t)]} \nbbP[{\cal C}_1(x, t)| {\cal E}_0(x, t)] \\
& \stackrel{(b)}{=} \frac{(l - \sigma) G_1(l,t) {\rm d}l}{\Phi_1(t)}
\end{align*}
where $(a)$ follows from the aforementioned assumption, and $(b)$ using the fact that $\nbbP[{\cal C}_1(x, t)| {\cal E}_0(x, t)] = 1$.

Now consider that the arriving rod $B_{\sigma/2}(x)$ sees one deposited rod in the neighborhood. Its arrival is in a gap of color $c_1$ only if the deposited rod is assigned color $c_2$. The probability of this event is $1/2$.
Following the similar principle as $n=0$, we write
\begin{align*}
\nbbP[{\cal C}_1(x, t, l)| {\cal E}_1(x, t)] &  = \nbbP[{\cal C}_1(x, t, l), {\cal C}_1(x, t)| {\cal E}_1(x, t)] \\ 
& = \nbbP[{\cal C}_1(x, t, l)| {\cal C}_1(x, t), {\cal E}_1(x, t)] \nbbP[{\cal C}_1(x, t)| {\cal E}_1(x, t)] \\ 
& = \frac{(l - \sigma) G_1(l,t) {\rm d}l}{\Phi_1(t)} \frac{1}{2}. \numberthis
\end{align*}
The event ${\cal E}_2(x, t)$ is more interesting compared to the previous cases. First, if the centers of both the deposited rods are not separated by a distance $\sigma$, then these two rods need to be assigned two different colors. 
Hence, 
\begin{align*}
\nbbP[{\cal C}_1(x,t,l)]| {\cal E}_{2}(x,t), \{\text{Admitted rods are less than}\ \sigma\ \text{apart}\}] = 0
\end{align*}
as the arrival is no longer in a gap of color $c_1$. 
Hence, the event we are interested in is that the centers of both the admitted rods are atleast $\sigma$ distance apart and both these rods are assigned color $c_2$. The probability that the two arrivals are at least $\sigma$ distance apart can be evaluated using order statistics and it comes out to be $5/18$. Further, the probability that these two rods are assigned color $c_2$ is $1/4$.
Overall, we write 
\begin{align*}
\nbbP[{\cal C}_1(x, t, l)| {\cal E}_2(x, t)] &  = \nbbP[{\cal C}_1(x, t, l), {\cal C}_1(x, t)| {\cal E}_2(x, t)] \\ 
& = \nbbP[{\cal C}_1(x, t, l)| {\cal C}_1(x, t), {\cal E}_2(x, t)] \nbbP[{\cal C}_1(x, t)| {\cal E}_2(x, t)] \\ 
& =  \frac{(l - \sigma) G_1(l,t) {\rm d}l}{\Phi_1(t)} \frac{1}{4} \frac{5}{18}. \numberthis
\label{eq:Prob_C1x_Cond_Enx}
\end{align*}
Using the law of total probability
\begin{align*}
\nbbP[{\cal C}_1(x, t, l)] = \frac{(l - \sigma) G_1(l,t) {\rm d}l}{\Phi_1(t)} \left(e^{-r_a 2\sigma t} + \frac{1}{2} (r_a 2\sigma t) e^{-r_a 2\sigma t} + \frac{1}{4}  \frac{5}{18} (r_a 2 \sigma t)^2 \frac{e^{-r_a 2\sigma t}}{2}\right). \numberthis
\label{eq:ProbC1x}
\end{align*}
Above expression for $\nbbP[{\cal C}_1(x, t, l)]$ is exact when 
\begin{align*}
\Phi_1(t) = e^{-r_a 2\sigma t} + \frac{1}{2} r_a 2\sigma t e^{-r_a 2\sigma t} + \frac{1}{4} \frac{5}{18} (r_a 2 \sigma t)^2 \frac{e^{-r_a 2\sigma t}}{2}.
\end{align*}
Since this is not the case, the result is an approximation whose accuracy is validated at the end of this section. 
}

{
To reach our final goal, we need to obtain $\nbbP[{\cal I}_1(x, t)| {\cal C}_1(x, t, l), {\cal E}_{n}(x, t)]$. Owing to the equi-probable random assignment of colors
\begin{align*}
 \nbbP[{\cal I}_1(x, t)| {\cal C}_1(x, t, l), {\cal E}_{n}(x, t)] = 
 \begin{cases}
 1/2 & n = 0 , \\
 1 & n = 1, 2,
 \end{cases} \numberthis
 \label{eq:Prob_I1x_Cond_C1x_Enx}
\end{align*}
where for $n=2$, we have the condition that both the arrivals are at least $\sigma$ distance apart.

Substituting the conditional probability expressions in \eqref{eq:Prob_I1x_Cond_C1x}, we get
\begin{align*}
\nbbP[{\cal I}_1(x, t)|{\cal C}_1(x, t, l)] = &  \frac{\frac{1}{2} e^{-r_a 2\sigma t} + \frac{1}{2} (r_a 2\sigma t) e^{-r_a 2\sigma t} + 5/144 (r_a 2\sigma t)^2 e^{-r_a 2 \sigma t}}{ e^{-r_a 2\sigma t} + \frac{1}{2} (r_a 2\sigma t) e^{-r_a 2\sigma t} + 5/144 (r_a 2\sigma t)^2 e^{-r_a 2\sigma t}} \\
= & \frac{1 + r_a 2\sigma t + 5/72(r_a (2 \sigma) t)^2}{2 + r_a 2\sigma t + 5/72(r_a (2 \sigma) t)^2}.
\end{align*}

}

Using all the intermediate steps described so far, we arrive at the following result to approximately characterizing the density of rods of a given color. 

\begin{prop}
For $K=2$, the density of rods of a given color $c_i$ is given as 
\begin{align*}
\rho_i(t)  = \int_{l \geq 0} G_i(l, t) {\rm d}l,
\end{align*}
where the time evolution of $G_i(l, t)$ is given as
\begin{align*}
\frac{\partial G_i(l, t)}{\partial t} =   
\begin{dcases}
   \left[- r_a (l - \sigma) G_i(l, t) + 2 r_a \int_{y = l+ \sigma}^{\infty} G_i(y, t) {\rm d}y\right] &  \\
   \quad \frac{1 + r_a 2\sigma t + 5/72(r_a (2 \sigma) t)^2}{2 + r_a 2\sigma t + 5/72 (r_a (2 \sigma) t)^2} & l \geq \sigma ,\\
       2 r_a \frac{1 + r_a 2\sigma t + 5/72(r_a (2 \sigma) t)^2}{2 + r_a 2\sigma t + 5/72 (r_a (2 \sigma) t)^2} \int_{y = l+ \sigma}^{\infty} G_i(y, t) {\rm d}y & l < \sigma.
\end{dcases}
\numberthis
\label{eq:Gltdt_N2}
\end{align*}
\end{prop}

Following results verify the accuracy of the approximation. In Fig.~\ref{fig:RSAvsBnP}, we present $G_i(l,t)$ as a function of $l$ for different $t$. We have considered the length of a rod as $\sigma = 1$. As evident from the figure, with increasing time, gaps of length $l < 1$ become relatively dominant of all the gaps. This result is also intuitive since only gaps of length $l < 1$ remain in the system as the system reaches the jamming limit.
In Fig.~\ref{fig:Density_Continuous_BC}, we present the density $\rho_i(t)$ of rods of a given color. 
We also observe that the simulations and approximated theory results are remarkably close. 

\begin{figure*}[!htb]
\centering
\begin{subfigure}{0.48\textwidth}
  \centering
  \includegraphics[width=1\linewidth]{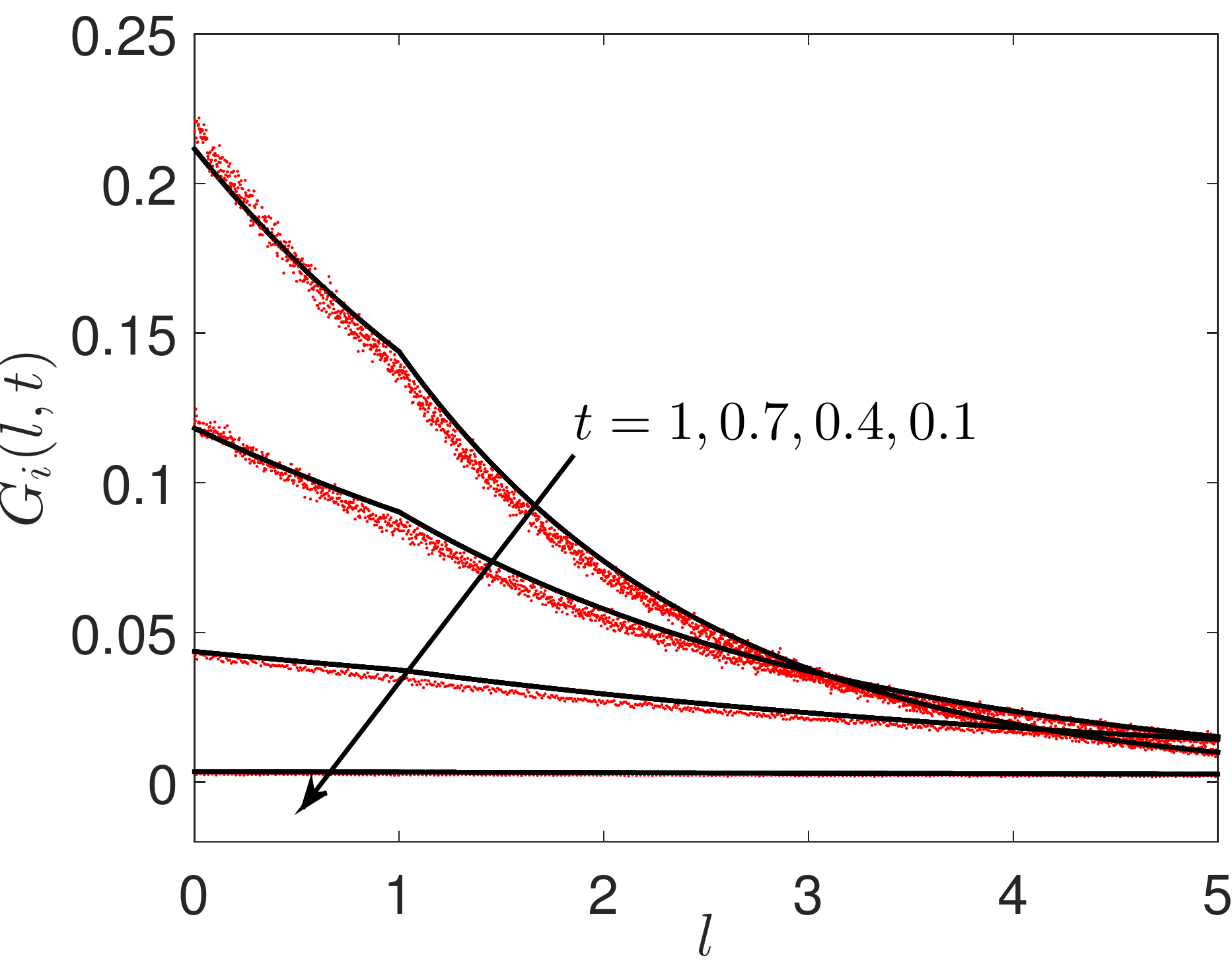}
\end{subfigure}
\begin{subfigure}{0.48\textwidth}
  \centering
  \includegraphics[width=1\linewidth]{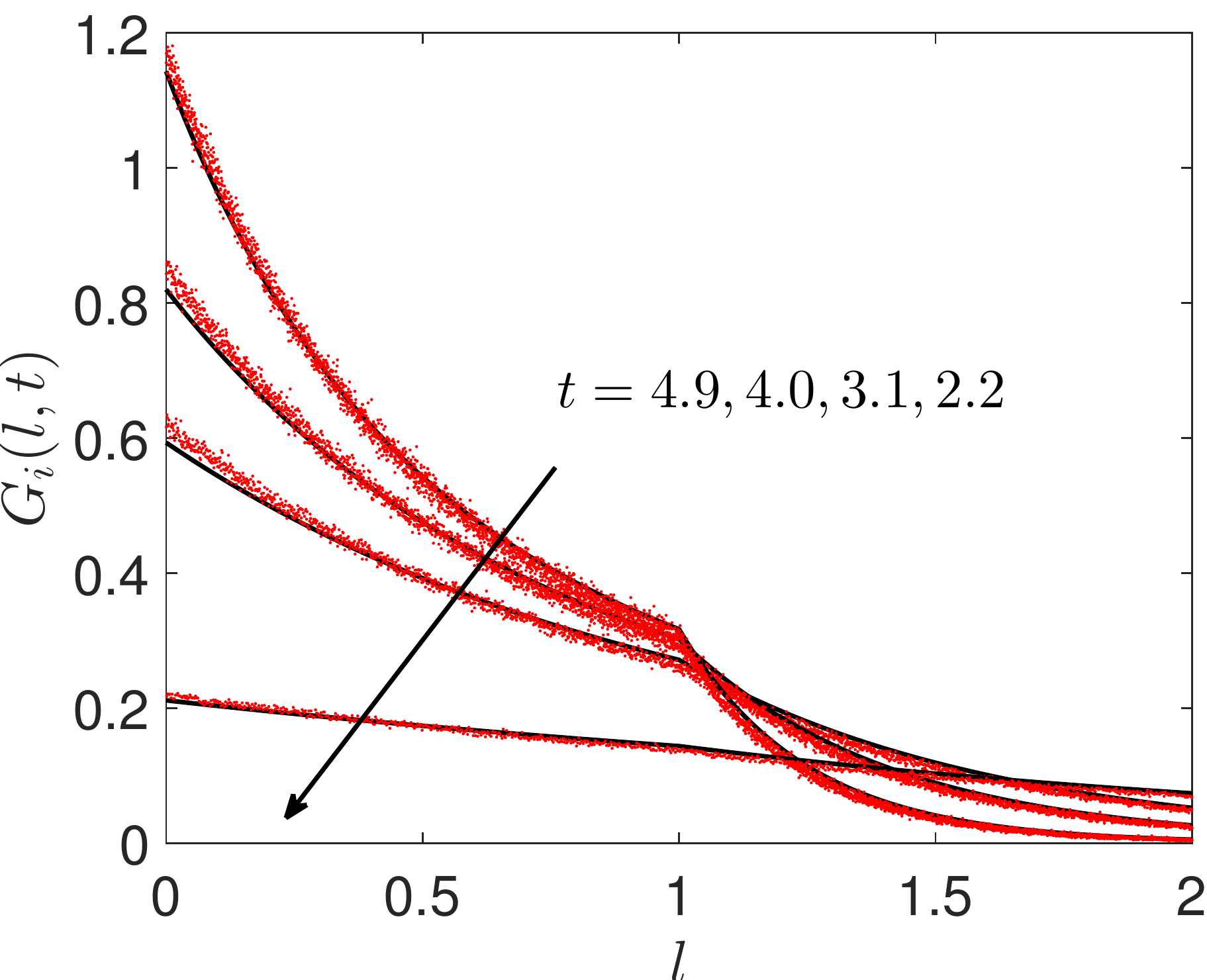}
\end{subfigure}%
\caption{\footnotesize  The evolution of gap density function for rods of a particular color as a function of gap length $l$ for $\sigma = 1$.  Solid lines and dotted markers represent theoretical approximation and simulations result, respectively.}
\label{fig:RSAvsBnP}
\end{figure*}

\begin{figure}[!htb]
  \centering
  \includegraphics[width=0.6\columnwidth]{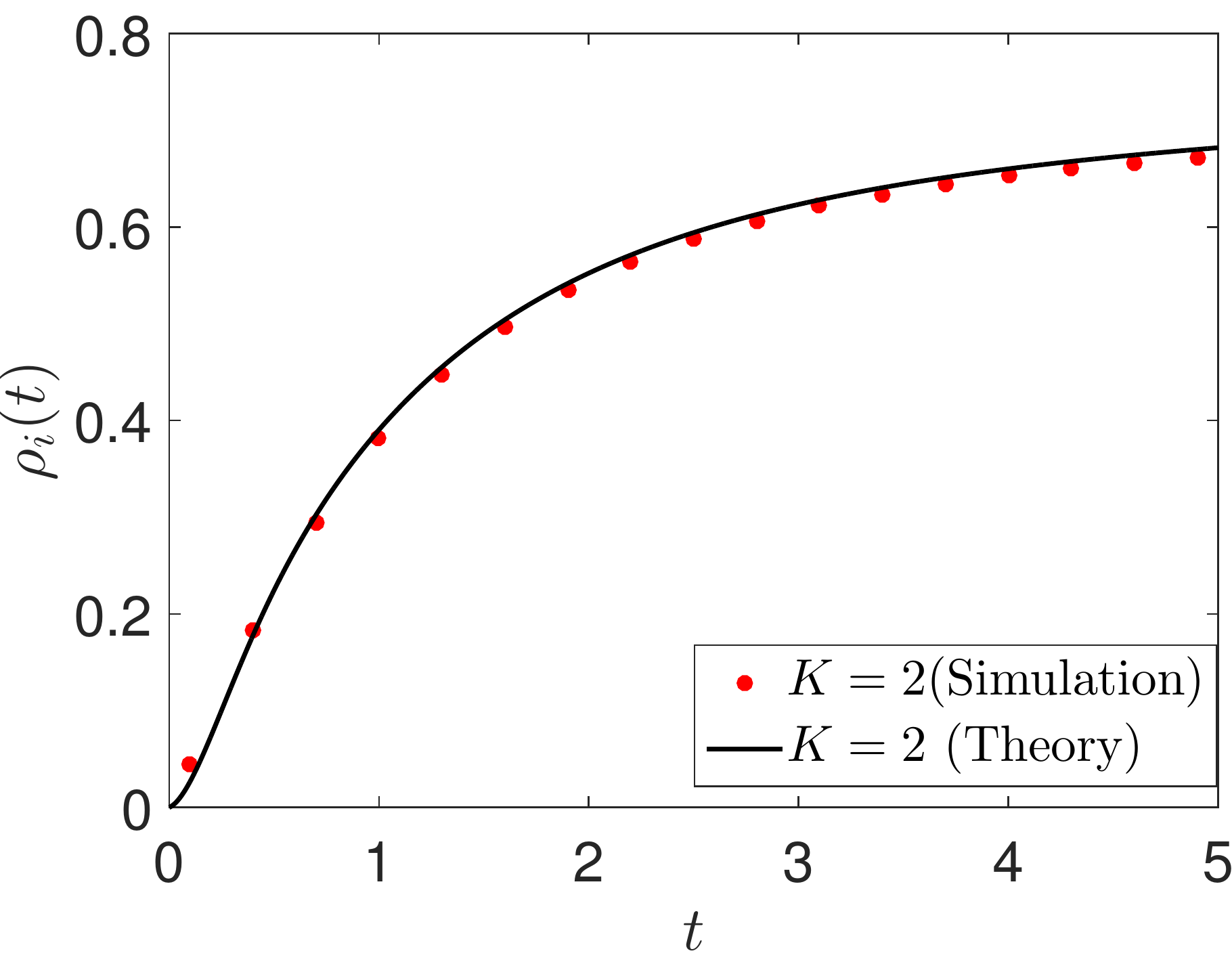}
  \caption{\footnotesize The evolution of density of rods of a particular color as a function of time $t$ for $\sigma = 1$.}
  \label{fig:Density_Continuous_BC}
\end{figure}

\subsection{Results for generic $K$}\label{subsec:genN}
Our next goal is to extend the previous approximation to $K \geq 2$ layers. However, capturing all the events mentioned in the previous subsection to characterize the rate equation for the gap density function becomes increasingly tedious as the number of layers increases. 
Therefore, to keep the numerical evaluation tractable, we make the following assumptions. 
The first assumption is the same as the approximation we have used for the previous approach that ignores the spatial dependence among prior arrivals beyond a certain range. 
\begin{assumption}
The admitted rods in the neighborhood $B_{\sigma}(x)$ of an arriving rod at $x$ are assumed to be deposited uniformly at random and independent of arrivals beyond $B_{\sigma}(x)$. Hence, these prior arrivals are assumed to follow Poisson process in $B_{\sigma}(x)$.
\end{assumption}

Further, if two prior arrivals in $B_{\sigma}(x)$ are separated by a distance $\sigma$, then there is non-zero probability that these two arrivals can be assigned the same color. 
However, considering this case exactly becomes cumbersome even for $K \geq 3$. Hence, we make the following assumption to make the rate equation for the gap density function tractable.
\begin{assumption}
If there are $m < K$ admitted rods in $B_{\sigma}(x)$, then these are assigned $m$ different colors irrespective of their relative distances.
\end{assumption}

With these assumptions, we propose following approximation to characterize  the evolution of density of the rods of color $c_i$.

\begin{prop}\label{prop:Kgen}
For a multilayer RSA process with $K$ colors, the density of rods of a given color $c_i$ is given as 
\begin{align*}
\rho_i(t)  = \int_{l \geq 0} G_i(l, t) {\rm d}l,
\end{align*}
where the time evolution of $G_i(l, t)$ is given as
\begin{align*}
\frac{\partial G_i(l, t)}{\partial t} =   
\begin{dcases}
  \left[- r_a (l - \sigma) G_i(l, t) + 2 r_a \int\limits_{y = l+ \sigma}^{\infty} G_i(y, t) {\rm d}y\right] \frac{\sum\limits_{n=0}^{K-1} \frac{(r_a 2\sigma t)^n}{n!}}{\sum\limits_{n=0}^{K-1} (K-n) \frac{(r_a 2\sigma t)^n}{n!}} , & l \geq \sigma \\
       2 r_a \frac{\sum\limits_{n=0}^{K-1} \frac{(r_a 2\sigma t)^n}{n!}}{\sum\limits_{n=0}^{K-1} (K-n) \frac{(r_a 2\sigma t)^n}{n!}} \int_{y = l+ \sigma}^{\infty} G_i(y, t) {\rm d}y, & l < \sigma.
\end{dcases}
\numberthis
\label{eq:Gltdt_N2}
\end{align*}
\end{prop}

\begin{proof}
The above proposition can be derived on the similar lines as the exposition of the $K=2$ case in the previous subsection. First, the rate of change equation for gap density function is the same as \eqref{eq:Gltdt_BC}.
Now the conditional probability expression in \eqref{eq:Gltdt_BC} can be expanded as 
\begin{align*}
& \frac{\nbbP[{\cal I}_i(x, t), {\cal C}_i(x, t, l)]}{\nbbP[{\cal C}_i(x, t, l)]} \\
= & \frac{\sum_{n\geq 0} \nbbP[{\cal I}_i(x, t)| {\cal C}_i(x,t,l), {\cal E}_{n}(x,t)] \nbbP[{\cal C}_i(x,t,l)| {\cal E}_{n}(x,t)] \nbbP[{\cal E}_{n}(x,t)]}{\nbbP[{\cal C}_i(x,t,l)]}. \numberthis
\label{eq:CondProbAp2}
\end{align*}
Using both the assumptions mentioned above, we write
\begin{align*}
 \nbbP[{\cal E}_{n}(x,t)] = e^{-r_a 2\sigma t}\sum\limits_{n=0}^{K-1} \frac{(r_a 2\sigma t)^n}{n!}, \quad 0 \leq n \leq K-1. 
\end{align*}
Further, on the similar lines as discussed in the previous section
\begin{align*}
\nbbP[{\cal C}_i(x,t,l)| {\cal E}_{n}(x,t)] =  \frac{(l - \sigma) G_i(l,t) {\rm d}l}{\Phi_i(t)} \frac{K-n}{K}, \quad  0 \leq n  \leq K-1. 
\end{align*}
Hence, using the law of total probability, we write
\begin{align*}
\nbbP[{\cal C}_i(x,t,l)] = \sum_{n=0}^{K-1} \frac{(l - \sigma) G_i(l,t) {\rm d}l}{\Phi_i(t)} \frac{K-n}{K}  \frac{(r_a 2\sigma t)^n}{n!} e^{-r_a 2\sigma t}.
\end{align*}
Moreover, due to equi-probable assignment of colors
\begin{align*}
\nbbP[{\cal I}_i(x, t)| {\cal C}_i(x,t,l), {\cal E}_{n}(x,t)] = \frac{1}{K-n}, \quad 0 \leq n \leq K-1. 
\end{align*}
Using the above four equations in \eqref{eq:CondProbAp2}, we get
\begin{align*}
\nbbP[{\cal I}_i(x, t)|{\cal C}_i(x, t, l)] = \frac{\sum\limits_{n=0}^{K-1} \frac{(r_a 2\sigma t)^n}{n!}}{\sum\limits_{n=0}^{K-1} (K-n) \frac{(r_a 2\sigma t)^n}{n!}}.
\end{align*}
The final result is obtained using the relationship between the gap density function and the density of rods of a particular color.
\end{proof}

The density result using the above proposition is presented in Fig.~\ref{fig:Density_apprx1b}. From the figure we observe that the simulations and the approximate theoretical results are remarkably close. Further, as expected the time required to reach the jamming limit increases as the number of layers increases. 

\begin{figure}[!htb]
  \centering
  \includegraphics[width=0.6\columnwidth]{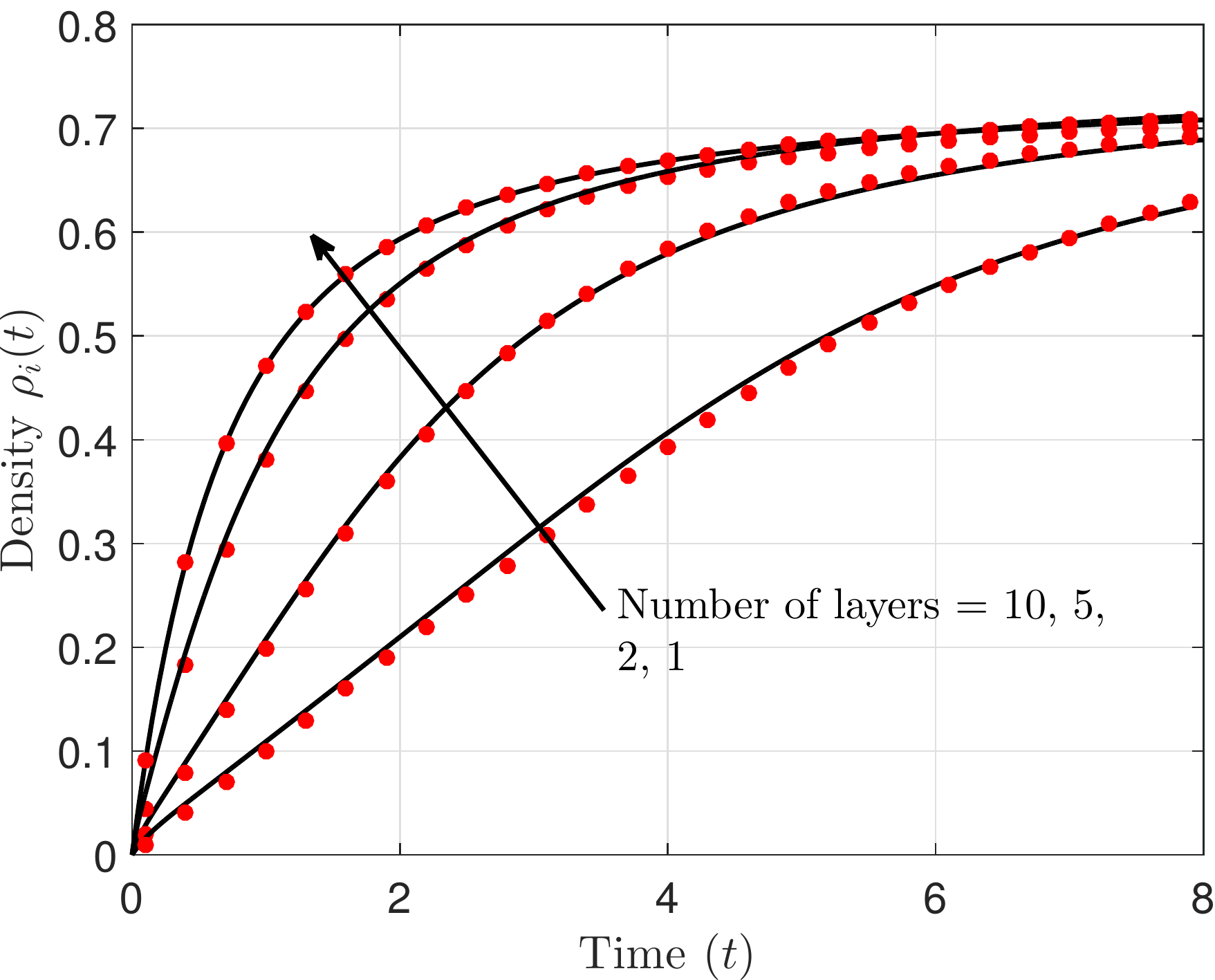}
  \caption{\footnotesize The evolution of density of rods of a particular color as a function of time $t$. The length of rods is $\sigma = 1$. Markers  and solid lines represent simulations and theoretical results, respectively.}
  \label{fig:Density_apprx1b}
\end{figure}

\section{Extension to 2D Multilayer RSA}\label{sec:Apprx2_2D}
In this section, we present the 2D version of the proposed multilayer RSA problem. We consider that circles with diameter $\sigma$ arrive uniformly at random in $\R^2$. Let there be $K$ colors in the system $\ncalK  = \{c_1, c_2, \ldots, c_K\}$ that are assigned to these circles based on the following rules:
\begin{enumerate}
\item An arriving circle that does not overlap with any of the admitted circles is assigned a color uniformly at random from $\ncalK$.
\item If the circle overlaps with $n < K$ colors, it is assigned a color uniformly at random from the rest of the colors. 
\item If the circle overlaps with all the colors, it is not admitted into the system. 
\end{enumerate}

\begin{figure}[!htb]
  \centering
  \includegraphics[width=0.9\columnwidth]{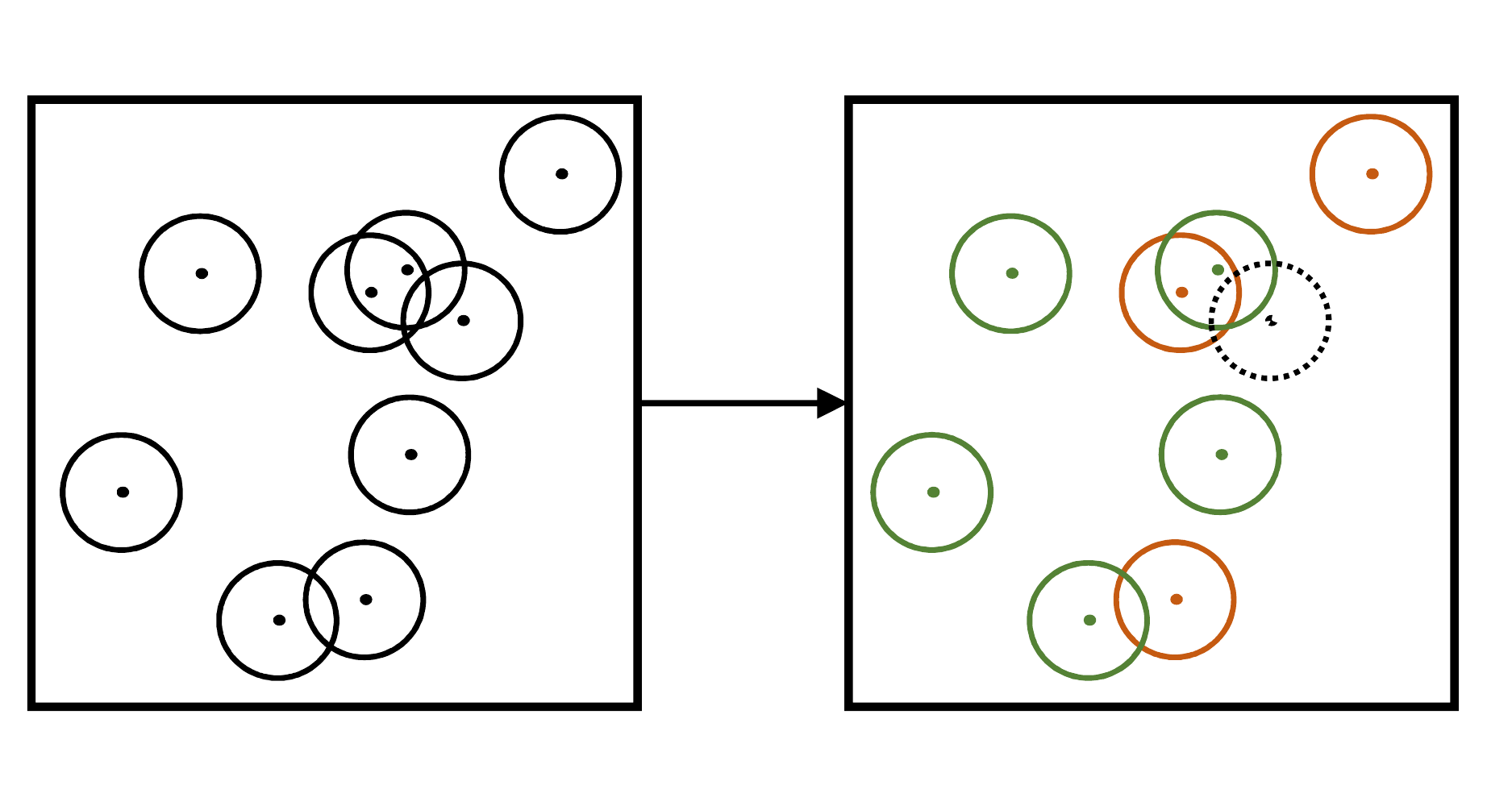}
  \caption{\footnotesize An illustration of the frequency band assignment process in a 2D wireless network. (Left) All the nodes that appear for transmission before a given time $t$. (Right) Nodes with the same color are assigned the same frequency band for transmission. Since there are two orthogonal frequency bands, only two out of three nodes with overlapping communication ranges are allowed to transmit. The node with a dotted circle remains silent.}
  \label{fig:Illus_2D}
\end{figure}

Based on the above rules, an illustrative example is given in Fig.~\ref{fig:Illus_2D} where we have considered $K=2$. In the left figure, all the arrivals before a given time $t$ are presented. From a communications network perspective, the centers represent communicating nodes, the range of each node is represented by a circle centered at the node. These nodes transmit to their respective receivers (not shown in the illustration) on an orthogonal frequency band out of the two available bands. In the right figure, nodes with the same color transmit on the same frequency band. Hence, interference is reduced among nodes that are within the communication range of each other. 

Similar to the 1D case, our goal is to obtain the density of circles of a given color. Since the exact solution to the problem is extremely difficult to obtain even in the monolayer case~\cite{schaaf1989,Schaaf1989Kin}, we resort to an approximation.
It is natural to consider an extension of either of the two approximation approaches developed for the 1D case.
As mentioned in Section~\ref{sec:Apprx2_1D}, the first approximation based on the iterative application of the monolayer RSA result is highly tractable, which makes it a promising candidate for extension to higher dimensions. Even though the second approximation based on the gap density function is slightly more accurate, its setup does not lend itself for a natural extension to higher dimensions.
Therefore, to obtain the density of circles of a given color, we rely on extending the iterative approximation approach. 
In the sequel, we present this result for the 2D multilayer RSA case.

\subsection{Approximate density characterization}\label{subsec:2DExtn}
Since this approach requires the known density result for monolayer RSA to be invoked repeatedly, for the sake of completeness, we first present this result from the literature~\cite{schaaf1989}. 
Consider a 2D monolayer RSA process that is obtained from circles of diameter $\sigma$ arriving uniformly at random at rate $r_a$ per unit area per unit time. 
In the following lemma, we present $\rho(t)$, the density of the admitted circles at time $t$.
\begin{lemma}\label{lem:RSA_Density}
The density $\rho(t)$  is obtained by solving the following differential equation with the initial condition $\rho(0) = 0$:
\begin{align*}
\int \frac{{\rm d}\rho(t)}{\phi(\kappa \rho(t))} = \frac{r_a}{\kappa} t + C, \numberthis
\label{eq:DiffEqun}
\end{align*}
where $\kappa = \frac{\pi \sigma^2}{4}$ is the area covered by a circle, $\kappa \rho(t)$ is the fraction of the area that is covered by the retained circles at time $t$, $\phi(\kappa \rho(t))$ is the probability that a circle arriving at an arbitrary location in $\R^2$ is retained at time $t$, and $C$ is the integration constant.
The series expansion of the retention probability in terms of density $\rho(t)$ is given as~\cite[Eq.~30]{schaaf1989}
\begin{align*}
\phi(\kappa \rho(t)) =  & 1 - 4 \pi \sigma^2 \rho(t) + \frac{\rho(t)^2}{2} \int_{\sigma}^{2 \sigma} 4 \pi r A_2(r) {\rm d}r + \frac{\rho(t)^3}{3} \int_{\sigma}^{2 \sigma} 2 \pi r A_2^2(r) {\rm d}r \\
& - S_3^{\tt eq} + O(\rho(t)^4), \numberthis
\label{eq:ExclusionProb}
\end{align*}
where $S_3^{\tt eq} = \frac{\rho(t)^3}{8} \pi \left(\sqrt{3} \pi - \frac{14}{3}\right)\sigma^6 + O(\rho(t)^4)$, $A_2(r)$ is the area of intersection of two circles of radius $\sigma$ whose centers are separated by distance $r$.
\end{lemma}
\begin{proof}
For the detailed proof of this lemma, please refer to~\cite{schaaf1989}. 
We just present the proof sketch here. 
Note that $\kappa\rho(t)$ is the fraction of area covered by the retained circles at time $t$.
Now, the rate of change of the fraction of the covered area depends on the number of arrivals $r_a {\rm d}t$ per unit area and the probability of an arrival being retained, which is given by $\phi\left(\kappa \rho(t)\right)$. Hence, 
\begin{align*}
 \frac{{\rm d} (\kappa\rho(t))}{{\rm d}t} = r_a \phi\left(\kappa \rho(t)\right). \numberthis
\label{eq:Langmuir_RSA}
\end{align*}
The expression for $\phi\left(\kappa \rho(t)\right)$ is derived in~\cite{schaaf1989} .
\end{proof}

The result of the above lemma is accurate up to a coverage of about 35\% by all the admitted circles. Using the knowledge of the asymptotic coverage of the 2D RSA process at the jamming limit, a unified equation for the retention probability is presented in~\cite{schaaf1989} that is accurate for the entire coverage range. This equation is given as
\begin{align*}
\phi_{\tt FIT}(\rho(t)) = (1 + b_1 x(t) + b_2 x(t)^2 + b_3 x(t)^3)(1-x(t)^3), \numberthis
\label{eq:FitFun}
\end{align*}
where $x(t) = \rho(t)/\rho(\infty)$ and $\rho(\infty)\kappa = 0.5474$ is the fraction of the area that is covered at the jamming limit as $t \rightarrow \infty$. The coefficients $b_1 = 0.8120, b_2 = 0.4258$ and $b_3=0.0716$ are obtained by matching the order of $\rho(t)$ in equations~\eqref{eq:ExclusionProb} and \eqref{eq:FitFun}.
Now the expression for $\rho(t)$ can be obtained by numerically solving the differential equation \eqref{eq:DiffEqun} using \eqref{eq:FitFun}.

As mentioned earlier, we use the same approach as the 1D multilayer RSA presented in Sec.~\ref{sec:Apprx2_1D} to approximate the density of circles of a given color. 
Let us extend the sequential color assignment process presented in Sec.~\ref{sec:Apprx2_1D} for 2D case, where an arriving circle is considered to be assigned $c_1$ before $c_2$ and so on.
Let $\tilde{\rho}_i(t)$ be the density of circles of $i$-th color under this sequential assignment scheme. 
In the following proposition we present the approximate result to estimate the density of circles of a given color for 2D multilayer RSA with the original random color assignment scheme.

\begin{prop}\label{prop:2DRSA}
The density of circles of a given color for 2D multilayer RSA with random color assignment scheme is given as 
\begin{align*}
\rho_i(t) = \frac{\sum_{k=1}^K \tilde{\rho}_k(t)}{K},
\end{align*}
where $\tilde{\rho}_k(t)$ is obtained by solving the monolayer RSA problem using Lemma~\ref{lem:RSA_Density} with adjusted rate of arrival per unit area for the $k$-th layer as $r_a - \sum_{i=1}^{k-1} \frac{\tilde{\rho}_i(t)}{t}$. 
\end{prop}

In Fig.~\ref{fig:Coverage_2D_mRSA}, we present the fraction of the total area covered by circles of a given color as a function of time.
From the figure, we see the approximated theoretical result are in close agreement with the Monte Carlo simulations result for different number of colors.

\begin{figure}[!htb]
  \centering
  \includegraphics[width=0.6\columnwidth]{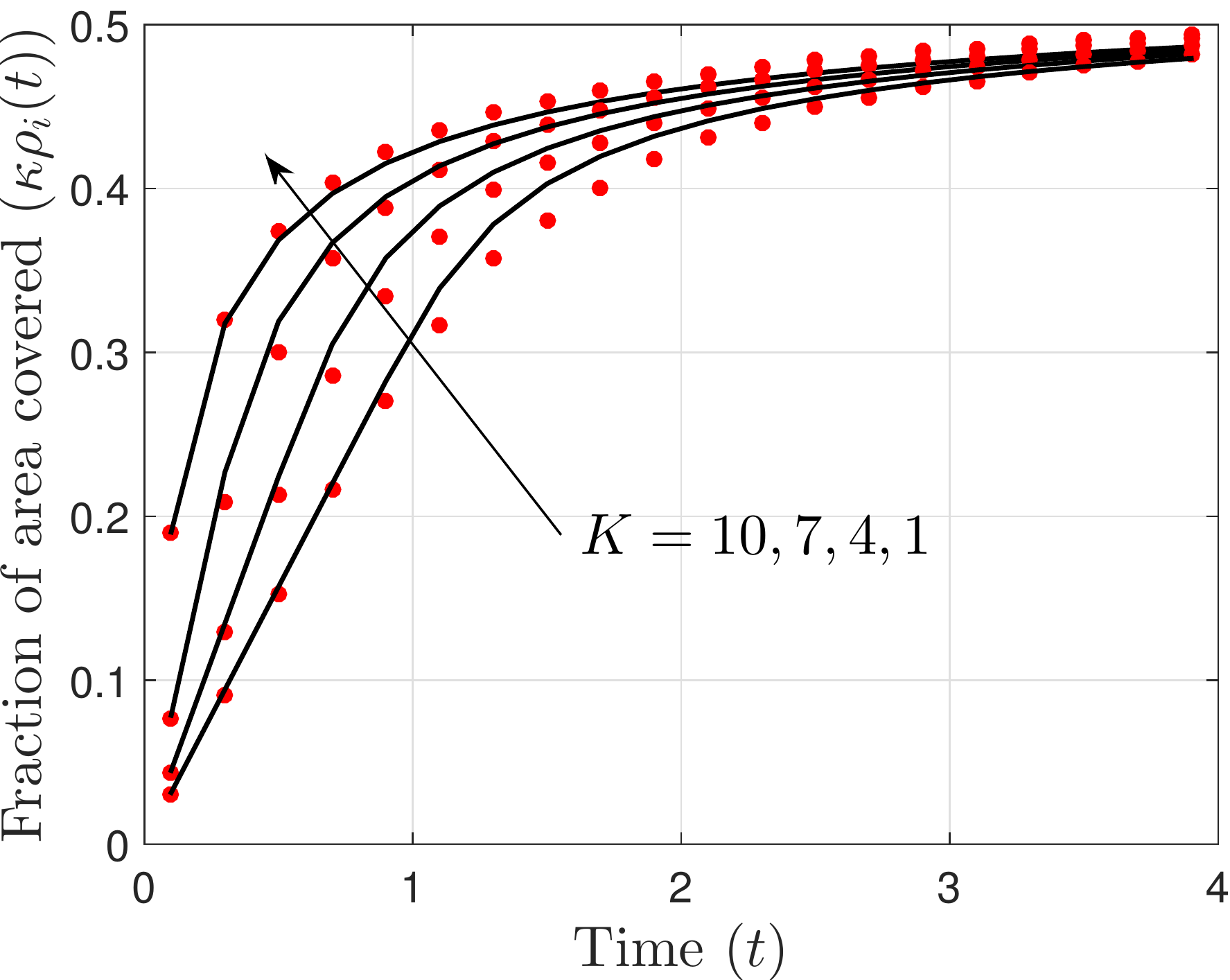}
  \caption{\footnotesize The fraction of the area covered by a circle of a particular color for a 2D multilayer RSA as a function of time $t$. Markers and solid lines represent simulations and theoretical results, respectively.}
  \label{fig:Coverage_2D_mRSA}
\end{figure}

\subsection{Application to wireless communication networks}
{In order to make a concrete connection of this work with wireless networks, we now present an application of the results derived in this section to a wireless local area network/Wi-Fi network. In these networks, the available orthogonal frequency bands or {\em channels} are limited, e.g., the Wi-Fi systems operating at 2.4 GHz have 11 channels. Hence, for data transmission, these channels are spatially reused by the access points (APs) throughout the network. The APs transmitting on the same channels are termed as {\em co-channel} APs. If two co-channel APs are in close proximity of each other, a significant amount of interference will degrade the performance of the users served by both the APs. To mitigate the effect of this co-channel interference, different dynamic channel assignment schemes have been investigated~\cite{Chieochan2010}.
For this specific example, we consider a distributed channel assignment scheme where an AP senses transmission on each channel and randomly selects one of the channels where the interfering signal strength from the closest AP on that channel is below a certain {\em sensing threshold}. 
As a consequence of this dynamic channel selection scheme, the co-channel APs ensure a minimum distance among themselves, which is termed as {\em inhibition distance} and is denoted by $d_{\tt inh}$. This distance depends on the transmit power of the APs, the propagation characteristics of the environment, and the sensing threshold. As an example, consider the popular single-slope path-loss model, $l(r) = r^{-\alpha}$, where $\alpha$ is the path-loss exponent. Further, assume that the transmit power of each node to be $P_t$ and the sensing threshold to be $I_{\tt th}$. In this case, we can express $d_{\tt inh} = l^{-1}\left(I_{\tt th}/P_t\right) = \left(I_{\tt th}/P_t\right)^{-1/\alpha}$.
In the wireless communication literature, the performance of this type of system has been studied for the case when there is a single frequency band in the network. 
In such a scenario, one popular spatial model that has been extensively used to model the locations of co-channel APs is the Mat{\'e}rn hard-core point process of type - II (MHPP-II)~\cite{nguyen2007,alfano2012new,matern2013spatial}\footnote{In the statistical physics literature this process is better known as the ghost RSA process~\cite{Torquato2006,Zhang2013}.}. The multilayer extension of the MHPP-II process is trivial, where the results for a given layer can be obtained using the single layer result with a scaled arrival rate, where the scale factor is $\frac{1}{K}$ for $K$ number of layers (frequency bands). Hence, in the wireless literature multilayer MHPP-II has not generated much interest. For modeling the Wi-Fi network with a single frequency band, although the MHPP-II is slightly more tractable than the RSA process, it underestimates the density of the co-channel APs that is accurately modeled by the RSA process. One consequence of the lower density of APs is an underestimation of total interference in the network that may lead to overly optimistic performance evaluation. Therefore, for the accurate performance analysis of the single frequency band case, the RSA process has also been considered in the wireless communications literature (cf.~\cite{busson2009,nguyen2012}), e.g., in~\cite{busson2009} the density result is derived through numerical simulations and in~\cite{nguyen2012} useful upper and lower bounds are presented for the generating functional of the RSA process. However, to the best of our knowledge, more practically relevant problem of multiple frequency bands has not yet been tackled from this perspective. Hence, the results derived in this section can be applied to evaluate the performance of such a system. With this background, next, we present an application of the 2D multilayer RSA result that provides useful network design guidelines for the aforementioned multiple frequency band Wi-Fi network.}

Consider, a Wi-Fi network where the AP locations are modeled as a homogeneous PPP $\Psi_t$ of density $\lambda_t$ per unit area. Let there be $K$ orthogonal channels in the system that are reused by the APs. Whenever an AP has data to transmit/receive with its associated user(s), it selects a suitable channel on which the data transmission will occur.
As discussed above, to avoid severe interference, the AP should select a channel such that the nearest co-channel AP will be at least $d_{\tt inh}$ distance apart. It can be easily argued that during a finite observation time window, the set of active APs follow the multilayer RSA process studied in this work.
Note that in case of the event that there is no available channel that satisfies the minimum distance criteria, the AP does not get to transmit/receive data with its associated users. From an operational point of view, this scenario is highly undesirable as users associated with this AP will get no service until it gets access to a channel.
Further, the probability of occurrence of this event becomes 1 as the system approaches the jamming limit. Therefore, it is necessary that the system operates well below the jamming limit such that a certain probabilistic guarantee can be made for a new AP to access a channel. This can be ensured by suitable selection of the system parameters such that the fraction of area covered by a set of co-channel APs is sufficiently below the jamming coverage of 0.5474. 
The parameters that can be tuned to achieve this objective are the transmit power of the APs and the sensing threshold. As mentioned earlier, the combined effect of changes to these parameters is directly captured by the inhibition distance $d_{\tt inh}$. 
Now, using Proposition~\ref{prop:2DRSA} along with the information on the number of channels $K$, inhibition distance $d_{\tt inh}$, and $\lambda_t$, we can determine the fraction of the total area that is covered by a set of co-channel APs.
In Fig.~\ref{fig:AppWireless}, we present the desired $d_{\tt inh}$ as a function of $\lambda_t$ to ensure that the fraction of area covered by a set of co-channel APs is 70\% of the jamming limit value. The plot is presented for systems operating at 2.4 GHz and 5 GHz that have 11 and 23 orthogonal channels, respectively. 
As observed from the figure, with increasing $\lambda_t$, $d_{\tt inh}$ should be reduced so that the target fraction of the area covered by a set of co-channel APs remains the same. Further, as expected, having a higher number of channels allows a larger $d_{\tt inh}$ for the same $\lambda_t$. From a deployment perspective, this result can be useful in adaptively selecting $d_{\tt inh}$ based on $\lambda_t$ that can vary based on the activity of the users, e.g. university campuses remain busy during the day time, but in the night time, the user activity drastically reduces, thereby providing the scope for the dynamic selection of $d_{\tt inh}$.
}

\begin{figure}[!htb]
  \centering
  \includegraphics[width=0.6\columnwidth]{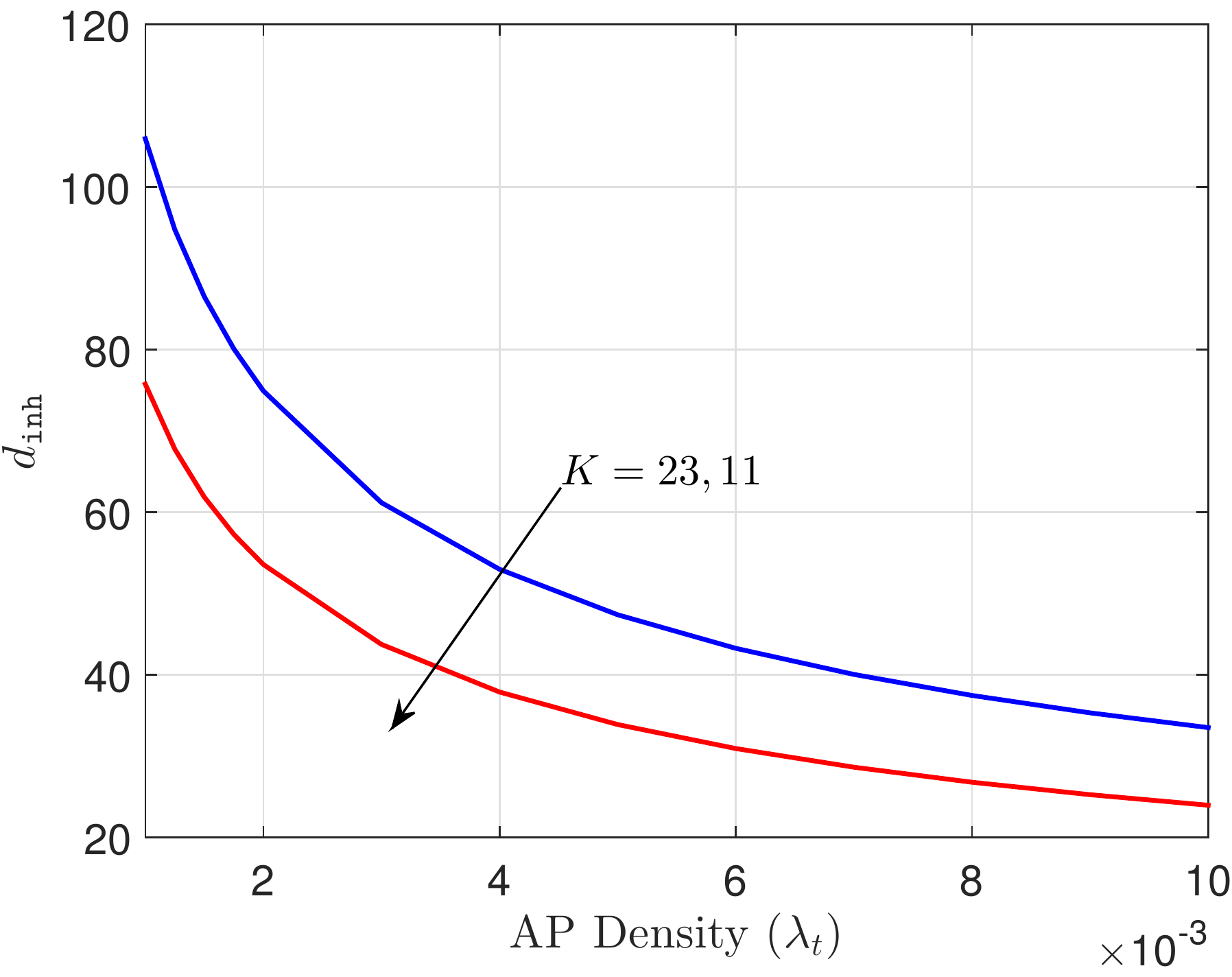}
  \caption{\footnotesize Inhibition distance $d_{\tt inh}$ as a function of AP density $\lambda_t$ to ensure that the fraction of covered area by a set of co-channel APs is 70\% of the jamming limit value. The plots correspond to the numbers of available orthogonal channels for Wi-Fi systems operating at 2.4 GHz and 5 GHz.}
  \label{fig:AppWireless}
\end{figure}

\section{Conclusion}\label{sec:concl}
In this article, we introduced a new variant of the multilayer RSA process that is inspired by the orthogonal resource sharing in wireless networks. 
For the 1D version of this process, we presented two useful approximations to obtain the density of deposited rods for a given layer. 
While our first approach is more amenable to numerical evaluation, the second approach is more accurate and provides useful information regarding the gap density function, which is an important statistical quantity to understand the kinetics of the RSA process.
We have also extended the first approximation to obtain the density of a given layer for the 2D version of this multilayer RSA process. 
Further, we have demonstrated the usefulness of the 2D result through one of its potential applications to model and analyze Wi-Fi networks.

There are many potential extensions of this work from the perspectives of both statistical physics and wireless communications.  
While the fraction of the space covered at the jamming limit is well-known for the monolayer RSA, the same is not true for this variant of the multilayer RSA, which is a promising direction for future work. For instance, it will be interesting to understand the fraction of space (i.e., length or area) that will be covered by the admitted particles (i.e., rods or circles) of all colors at the jamming limit. A similar question can be asked to characterize the fraction of space that is jointly covered by particles of more than one color that also has key applications to wireless networks, such as in cloud radio access networks and cell-free networks. 
Another possible future direction is to further extend this variant of the multilayer RSA to higher dimensions. 
While in this work we have focused on the first-order statistics, namely the density of the process, derivations of higher-order statistics, such as the pair correlation function, third-moment density, are natural next steps in further investigation of this process. 

{Another interesting extension of this work is the consideration of non-isotropic inhibition regions that is inspired by the notion of soft-connectivity in the network percolation theory~\cite{dettmann2016random}. In this work, we have considered impenetrable hard rods/circles. This accurately captures inhibition regions in the scenario where the interference is modeled using a generic path-loss function. However, since wireless signals additionally suffer from random fluctuations due to shadowing and fading, it is natural to consider soft non-isotropic inhibition regions as an extension to this work. In this case, an arriving node is admitted into the system with a certain probability based on its relative distance from the already admitted nodes. Depending on this probabilistic inhibition function, the evolution of the density and the jamming properties of this soft RSA process are going to be different from the multilayer RSA process considered in this work. Hence, it is worthwhile to study the properties of this {\em soft multilayer RSA process}.}

From the communications network perspective, the 2D results of this work can be extended to study frequency reuse in cellular networks or the reuse of pilot sequences in cell-free communications networks~\cite{parida2021pilot}. Further, the 1D results can be extended to analyze multimedia resource reservation systems with multiple orthogonal channels~\cite{coffman1998}.


%
%

\bibliographystyle{IEEEtran}
\bibliography{RSA_Related_Lit}   

\end{document}

%% file: notation.tex
\def\nb0{{\mathbf{0}}}
\def\nb1{{\mathbf{1}}}

\def\ncalK{{\mathcal{K}}}

\def\nbbP{{\mathbb{P}}}

\newtheorem{lemma}{Lemma}
\newtheorem{ndef}{Definition}
\newtheorem{prop}{Proposition}
\newtheorem{assumption}{Assumption}
	
\def\R{\mathbb{R}}












%% file: RSA_JSP_IEEE_format.bbl
\begin{thebibliography}{10}
\providecommand{\url}[1]{#1}
\csname url@samestyle\endcsname
\providecommand{\newblock}{\relax}
\providecommand{\bibinfo}[2]{#2}
\providecommand{\BIBentrySTDinterwordspacing}{\spaceskip=0pt\relax}
\providecommand{\BIBentryALTinterwordstretchfactor}{4}
\providecommand{\BIBentryALTinterwordspacing}{\spaceskip=\fontdimen2\font plus
\BIBentryALTinterwordstretchfactor\fontdimen3\font minus
  \fontdimen4\font\relax}
\providecommand{\BIBforeignlanguage}[2]{{%
\expandafter\ifx\csname l@#1\endcsname\relax
\typeout{** WARNING: IEEEtran.bst: No hyphenation pattern has been}%
\typeout{** loaded for the language `#1'. Using the pattern for}%
\typeout{** the default language instead.}%
\else
\language=\csname l@#1\endcsname
\fi
#2}}
\providecommand{\BIBdecl}{\relax}
\BIBdecl

\bibitem{flory1939}
P.~J. Flory, ``Intramolecular reaction between neighboring substituents of
  vinyl polymers,'' \emph{J. Am. Chem. Soc.}, vol.~61, no.~6, pp. 1518--1521,
  1939.

\bibitem{Renyi1958}
A.~R{\'e}nyi, ``On a one-dimensional problem concerning random space filling,''
  \emph{Publications of the Mathematical Institute of the Hungarian Academy of
  Sciences}, vol.~3, pp. 109--127, 1958.

\bibitem{talbot2000car}
J.~Talbot, G.~Tarjus, P.~Van~Tassel, and P.~Viot, ``From car parking to protein
  adsorption: an overview of sequential adsorption processes,'' \emph{Colloids
  Surf. A: Physicochem. Eng. Asp.}, vol. 165, no. 1-3, pp. 287--324, 2000.

\bibitem{coffman1998}
E.~G. Coffman~Jr, L.~Flatto, P.~Jelenkovi{\'c}, and B.~Poonen, ``Packing random
  intervals on-line,'' \emph{Algorithmica}, vol.~22, no.~4, pp. 448--476, 1998.

\bibitem{Andrews2011}
J.~G. Andrews, F.~Baccelli, and R.~K. Ganti, ``A tractable approach to coverage
  and rate in cellular networks,'' \emph{IEEE Trans. on Commun.}, vol.~59,
  no.~11, pp. 3122--3134, 2011.

\bibitem{DhiGanJ2012}
H.~S. Dhillon, R.~K. Ganti, F.~Baccelli, and J.~G. Andrews, ``Modeling and
  analysis of {K}-tier downlink heterogeneous cellular networks,'' \emph{IEEE
  Journal on Sel. Areas in Commun.}, vol.~30, no.~3, pp. 550--560, 2012.

\bibitem{mankar2020}
P.~D. Mankar, P.~Parida, H.~S. Dhillon, and M.~Haenggi, ``Distance from the
  nucleus to a uniformly random point in the 0-cell and the typical cell of the
  {P}oisson--{V}oronoi tessellation,'' \emph{J. Stat. Phys.}, vol. 181, no.~5,
  pp. 1678--1698, 2020.

\bibitem{pineda2007}
E.~Pineda and D.~Crespo, ``Temporal evolution of the domain structure in a
  {P}oisson-{V}oronoi transformation,'' \emph{J. Stat. Mech. Theory Exp.}, vol.
  2007, no.~06, p. P06007, 2007.

\bibitem{koufos2019}
K.~Koufos and C.~P. Dettmann, ``Distribution of cell area in bounded {P}oisson
  {V}oronoi tessellations with application to secure local connectivity,''
  \emph{J. Stat. Phys.}, vol. 176, no.~5, pp. 1296--1315, 2019.

\bibitem{dhillon2020}
H.~S. Dhillon and V.~V. Chetlur, ``{P}oisson line {C}ox process: Foundations
  and applications to vehicular networks,'' \emph{Synth. Lect. Learn. Networks
  Algo.}, vol.~1, no.~1, pp. 1--149, 2020.

\bibitem{chetlur2020}
V.~V. Chetlur, H.~S. Dhillon, and C.~P. Dettmann, ``Shortest path distance in
  manhattan {P}oisson line {C}ox process,'' \emph{J. Stat. Phys.}, vol. 181,
  no.~6, pp. 2109--2130, 2020.

\bibitem{bartelt1990}
M.~Bartelt and V.~Privman, ``Kinetics of irreversible multilayer adsorption:
  One-dimensional models,'' \emph{J. Chem. Phys.}, vol.~93, no.~9, pp.
  6820--6823, 1990.

\bibitem{bartelt1991}
------, ``Kinetics of irreversible monolayer and multilayer adsorption,''
  \emph{Int. J. Mod. Phys. B}, vol.~5, no.~18, pp. 2883--2907, 1991.

\bibitem{krapivsky1992}
P.~Krapivsky, ``Kinetics of multilayer deposition: Models without screening,''
  \emph{J. Chem. Phys.}, vol.~97, no.~3, pp. 2134--2138, 1992.

\bibitem{van1997}
P.~Van~Tassel and P.~Viot, ``An exactly solvable continuum model of multilayer
  irreversible adsorption,'' \emph{EPL}, vol.~40, no.~3, p. 293, 1997.

\bibitem{yang1998}
S.~Yang, P.~Viot, and P.~R. Van~Tassel, ``Generalized model of irreversible
  multilayer deposition,'' \emph{Phys. Rev. B}, vol.~58, no.~3, p. 3324, 1998.

\bibitem{Fleurke2007}
S.~Fleurke and H.~Dehling, ``The sequential frequency assignment process,'' in
  \emph{12th WSEAS International Conference on Applied Mathematics}, 2007, pp.
  280--285.

\bibitem{fleurke2009second}
S.~Fleurke and C.~K{\"u}lske, ``A second-row parking paradox,'' \emph{J. Stat.
  Phys.}, vol. 136, no.~2, pp. 285--295, 2009.

\bibitem{fleurke2010multilayer}
------, ``Multilayer parking with screening on a random tree,'' \emph{J. Stat.
  Phys.}, vol. 139, no.~3, pp. 417--431, 2010.

\bibitem{Fleurke2014}
S.~R. Fleurke and A.~C.~D. van Enter, \emph{Analytical Results for a Small
  Multiple-Layer Parking System}.\hskip 1em plus 0.5em minus 0.4em\relax Cham:
  Springer International Publishing, 2014, pp. 43--53.

\bibitem{bonnier1994}
B.~Bonnier, D.~Boyer, and P.~Viot, ``Pair correlation function in random
  sequential adsorption processes,'' \emph{J. Phys. Math. Gen.}, vol.~27,
  no.~11, p. 3671, 1994.

\bibitem{schaaf1989}
P.~Schaaf and J.~Talbot, ``Surface exclusion effects in adsorption processes,''
  \emph{J. Chem. Phys.}, vol.~91, no.~7, pp. 4401--4409, 1989.

\bibitem{Schaaf1989Kin}
------, ``Kinetics of random sequential adsorption,'' \emph{Phys. Rev. Lett.},
  vol.~62, pp. 175--178, Jan 1989.

\bibitem{Chieochan2010}
S.~Chieochan, E.~Hossain, and J.~Diamond, ``Channel assignment schemes for
  infrastructure-based 802.11 {WLAN}s: A survey,'' \emph{IEEE Commun. Surveys
  and Tutorials}, vol.~12, no.~1, pp. 124--136, 2010.

\bibitem{nguyen2007}
H.~Q. Nguyen, F.~Baccelli, and D.~Kofman, ``A stochastic geometry analysis of
  dense {IEEE} 802.11 networks,'' in \emph{Proc., IEEE Intl. Conf. on Commun.
  (ICC)}.\hskip 1em plus 0.5em minus 0.4em\relax IEEE, 2007, pp. 1199--1207.

\bibitem{alfano2012new}
G.~Alfano, M.~Garetto, and E.~Leonardi, ``New directions into the stochastic
  geometry analysis of dense {CSMA} networks,'' \emph{IEEE Trans. Mobile
  Computing}, vol.~13, no.~2, pp. 324--336, 2012.

\bibitem{matern2013spatial}
B.~Mat{\'e}rn, \emph{Spatial variation}.\hskip 1em plus 0.5em minus 0.4em\relax
  Springer Science \& Business Media, 2013, vol.~36.

\bibitem{Torquato2006}
S.~Torquato and F.~H. Stillinger, ``Exactly solvable disordered sphere-packing
  model in arbitrary-dimensional euclidean spaces,'' \emph{Phys. Rev. B},
  vol.~73, p. 031106, Mar 2006.

\bibitem{Zhang2013}
G.~Zhang and S.~Torquato, ``Precise algorithm to generate random sequential
  addition of hard hyperspheres at saturation,'' \emph{Phys. Rev. B}, vol.~88,
  p. 053312, Nov 2013.

\bibitem{busson2009}
A.~Busson, G.~Chelius, and J.-M. Gorce, ``Interference modeling in {CSMA}
  multi-hop wireless networks,'' Ph.D. dissertation, INRIA, 2009.

\bibitem{nguyen2012}
T.~V. Nguyen and F.~Baccelli, ``On the spatial modeling of wireless networks by
  random packing models,'' in \emph{IEEE INFOCOM}, 2012, pp. 28--36.

\bibitem{dettmann2016random}
C.~P. Dettmann and O.~Georgiou, ``Random geometric graphs with general
  connection functions,'' \emph{Phys. Rev. B}, vol.~93, no.~3, p. 032313, 2016.

\bibitem{parida2021pilot}
P.~Parida and H.~S. Dhillon, ``Pilot assignment schemes for cell-free massive
  {MIMO} systems,'' \emph{arXiv:2105.09505}, 2021.

\end{thebibliography}
